\documentclass[11pt]{article}

\usepackage{amssymb}
\usepackage{amsthm}
\usepackage{graphicx}
\usepackage{xcolor}
\usepackage{amsmath}
\usepackage[ruled, vlined, linesnumbered, noresetcount]{algorithm2e}

\begin{document}

\SetKwFor{ForAll}{forall}{do}{end forall}

\newtheorem{ru}{Rule}
\newcommand{\vd}[1]{\nn(#1)}
\newcommand{\vpd}[1]{\nnp(#1)}
\newcommand{\vnd}[1]{\nnn(#1)}
\newcommand{\mynote}[1]{\textcolor{red}{Check: {#1}}}
\newcommand{\rednote}[1]{\textcolor{red}{{#1}}}
\newcommand{\remove}[1]{}
\newcommand{\Domination}{{{\sf Dom}}}
\newcommand{\Clique}{{{\sf Clique}}}
\newcommand{\Coloring}{{{\sf Color}}}
\newcommand{\OPT}{{\tt OPT}}
\def \t {{{t}}}
\def\V{{\cal V}}
\def\nd{{nd}}
\def\dnd{{ itp}}
\def\tD{{\tilde{D}}}
\def\tC{{\tilde{C}}}
\def\I{{\cal I}}
\def\Ho{{H^{(0)}}}
\def\H1{{H^{(1)}}}
\def\Hii{{H^{(i)}}}
\def\Hi1{{H^{(i-1)}}}
\def\Hd{{H^{(d)}}}

\def \Eq {{Equitable Coloring}}
\def \Bin {{Bin-Packing}}

\def \argmin{arg\,min}
\def \argmax{arg\,max}

\newcommand{\com}[1]{$O(#1+poly(n))$}

\newtheorem{proposition}{Proposition}
\newtheorem{corollary}{Corollary}
\newtheorem{fact}{Fact}
\newtheorem{theorem}{Theorem}
\newtheorem{example}{Example}
\newtheorem{definition}{Definition}
\newtheorem{claim}{Claim}
\newtheorem{problem}{Problem}
\newtheorem{lemma}{Lemma}
\newtheorem{remark}{Remark}
\newtheorem{proper}{Property}

\date{}
\pagestyle{plain}

\title{Iterated Type Partitions
}
%
%

\author{ G. Cordasco \\ \small University of Campania ``L.Vanvitelli'', Italy
 \and 
 L. Gargano \\ \small University of Salerno, Italy\\
\and
 A. A. Rescigno \\ \small University of Salerno, Italy}

\maketitle              
\begin{abstract}
This paper deals with the complexity of some natural graph problems when parametrized by {measures that are restrictions of} clique-width, such as modular-width and neighborhood diversity.
The main contribution of this paper is to introduce a novel parameter, called iterated type partition, that 
can be computed in polynomial time and nicely places between modular-width and neighborhood diversity. 
We prove that the Equitable Coloring problem is W[1]-hard when parametrized by the iterated type partition. 
This result extends to  modular-width, answering an open question about the possibility to have FPT algorithms for Equitable Coloring when parametrized by modular-width. 
On the contrary, we show that the Equitable Coloring problem is instead FPT when parameterized by neighborhood diversity.
Furthermore, we present simple and fast FPT algorithms parameterized by iterated type partition that provide optimal solutions for several graph problems; in particular this paper presents algorithms for the Dominating Set, the Vertex Coloring and the Vertex Cover problems. While the above problems are already known to be FPT with respect to modular-width, the novel algorithms are both  simpler and more efficient: For the Dominating set and Vertex Cover problems, our algorithms output an optimal set in time \com{2^t}, while for the Vertex Coloring problem, our algorithm outputs an optimal set in time \com{t^{2.5t+o(t)}\log n},
where $n$ and $t$ are the size and the iterated type partition of the input graph, respectively.\\ 		

\noindent
\textbf{keywords:} Parameterized Complexity, Fixed-parameter tractable algorithms,  W[1]-hardness, Neighborhood Diversity, Modular-width.
\end{abstract}

\section{Introduction}
Some NP-hard problems can be solved by algorithms that are exponential only in the size of a fixed parameter while they are polynomial in the size of the input. Such  problems are   called  fixed-parameter tractable, because the problem can be solved efficiently for small values of the  parameter  \cite{DF,N}. Formally, a parameterized problem with input size $n$ and  parameter $\t$ is called {\em fixed parameter tractable (FPT)} if it can be solved in time $f(\t) \cdot n^c$, where $f$ is a function only depending on $\t$ and $c$ is a constant.


An important quality of a parameter is that is is easy to compute.
Unfortunately there are several parameters  whose 
computation is an  NP-hard problem.  
As an example computing treewidth, rankwidth, and vertex cover are all NP-hard problems but 
they are computable in FPT time when their respective parameters are bounded;
moreover,  the parameterized complexity of computing the clique-width of a graph exactly is still an open problem \cite{DK}.

We     start from  two recently introduced parameters:
modular-width \cite{GLO}  and neighborhood diversity \cite{L}.
Both parameters    received  
 much attention \cite{ALM+,CBFGR,CGRV,CDP,DKT16,FGKKK,FLMT,G,GR,CGRV18} also 
due to their   property  of being computable  in polynomial time 
\cite{GLO,L}.  

As the main contribution of this paper we introduce a novel parameter
called  Iterated Type Partition, which nicely places between the two above parameters
 and allows to obtain new   algorithms and hardness results.

\subsection{Modular-width}\label{sec-mw}
The notion of modular decomposition of graphs was introduced by Gallai
in \cite{Gallai}, as a tool to define  hierarchical decompositions of graphs. It has been recently considered in \cite{GLO} to define the modular-width parameter in the area of parameterized computation.

{Consider  graphs   obtainable by an algebraic expression that uses 
the 
operations:}
\begin{itemize}
\item[1)] Creation of  an isolated vertex.
\item[2)] Disjoint union of 2 graphs, 
 i.e.,  the
graph with vertex set $V (G_1 ) \cup V ( G_2 )$ and edge set $E( G_1 )\cup E( G_2 )$.
\item[3)] Complete join of 2 graphs, 
i.e., 
the graph with vertex set $V (G_1 ) \cup V ( G_2 )$ and edge set 
$E( G_1 )\cup E( G_2 )\cup \{ (v, w)\  :\  v \in V ( G_1 ),\  w\in  V(G_2) \}$.
\item[4)] Substitution operation $G(G_1,\dots,G_m)$ 
 of the vertices $v_1,\dots,v_m$  of $G$ by
the  modules  $G_1,\dots,G_m$, i.e.,  the graph with vertex
set
$\bigcup_{1\leq \ell\leq m} V(G_\ell)$ and edge set

\centerline{
$\bigcup_{1\leq \ell \leq  m} E(G_\ell)\cup \{(u, v)\  : \ u \in V (G_i), \ v \in
V (G_j ), \  (v_i, v_j) \in E(G) \}. $
}
\end{itemize}
\smallskip
\remove{Let $A$ be an algebraic expression that uses only the operations 1)--4). The width of $A$  is 
defined as the maximum number of operands used by any occurrence of the operation 4) in $A$. The \textit{modular-width} of a
graph $G$, denoted $mw(G)$, is the least integer $m$ such that $G$ can be obtained
from such an algebraic expression of width at most $m$. It is well-known that  an algebraic
expression of width $mw(G)$ can be constructed in linear time \cite{TCH+}.}
{As defined in \cite{GLO}, the \textit{modular-width} of a
graph $G$, denoted $mw(G)$, is the least integer $m$ such that $G$ can be obtained by using  only the operations 1)--4) (in any number and order) and where each operation 4) has at most $m$ modules.}

\subsection{Neighborhood diversity}\label{sec-nd}
Given a graph $G=(V,E)$,  two nodes $u,v\in V$  have the same  
{\em type} iff $N(v) \setminus \{u\} = N(u) \setminus \{v\}$.
The {\em neighborhood diversity} of a graph $G$, introduced by Lampis in \cite{L} and denoted by {\em \nd}$(G)$, is the minimum number $\t$ of sets in a partition  $V_1,V_2, \ldots, V_\t$, of the node set $V$, 
such that all the nodes in  $V_i$ have the same type, 
for  $i\in [t]\footnote{
For a positive integer $n$, we use 
$[n]$ to denote the set of the first $n$ integers, that is $[n]=\{1,2, \ldots, n\}$.}$.
\\
The family 
$\V=\{V_1,V_2, \ldots, V_\t\}$ is called  the {\em type partition} of $G$.

Let $G=(V,E)$ be a graph with type partition  $\V=\{V_1,V_2, \ldots, V_\t\}$. 
By  definition, 
 each $V_i$ induces either a {\em clique} or an {\em independent set} in $G$.
We treat singleton sets in the type partition as cliques.
For each  $V_i,V_j\in \V$, we get that  either each node in 
$V_i$ is a neighbor of each node in $V_j$ or no node in $V_i$ 
has a neighbor  in $V_j$.
Hence, between each pair $V_i,V_j\in \V$, there is either a complete bipartite graph  
or  no edges at all.

\smallskip

Starting from a graph $G$ and its type partition $\V=\{V_1, \ldots, V_\t\}$,  
we can see each  element  of $\V$ as a {\em metavertex} of a new graph $H$,  called  the {\em type graph} of $G$,  with  \\
-   $V(H)=\{1,2,\cdots,\t\}$ \\
-  $E(H)=\{(x,y) \ | \ x \neq y \mbox{ and for each $u\in V_x$, $v\in V_y$  it holds that $(u,v)\in E(G)$ }\}$.
\\
We say that $G$ is a {\em base graph}  if it matches its type graph, 
that is, the type partition of $G$ consists of singletons, each representing 
a  node in $V(G)$, and $nd(G)=|V(G)|$. 

\smallskip
We introduce a new graph parameter, which  generalizes   neighborhood diversity.
Given a graph $G$, the {\em Iterated Type Partition} of  $G$ 
is defined  by iteratively constructing type graphs until a base graph is obtained. 
\begin{definition}\label{deep-nd}
Given a graph $G=(V,E)$,  let $\Ho=G$ and 
 $\Hii$ denote the  type graph of $\Hi1$, for $i\geq 1$.  
Let $d$ be the smallest  integer such that $\Hd$ is a base graph.
The {\em iterated type partition} of $G$, denoted by $\dnd(G)$, 
is the number of nodes of  $\Hd$.
The sequence of graphs $\Ho=G, \H1, \cdots, \Hd$ is called the  {\em type graph sequence} of $G$ and $\Hd$ is  denoted as the base graph of $G$.
\end{definition}
An example of a graph and its type graph sequence is given in Fig. \ref{fig:img1}. 
\begin{figure}[tb!]
	\centering
	\includegraphics[width=0.9\textwidth]{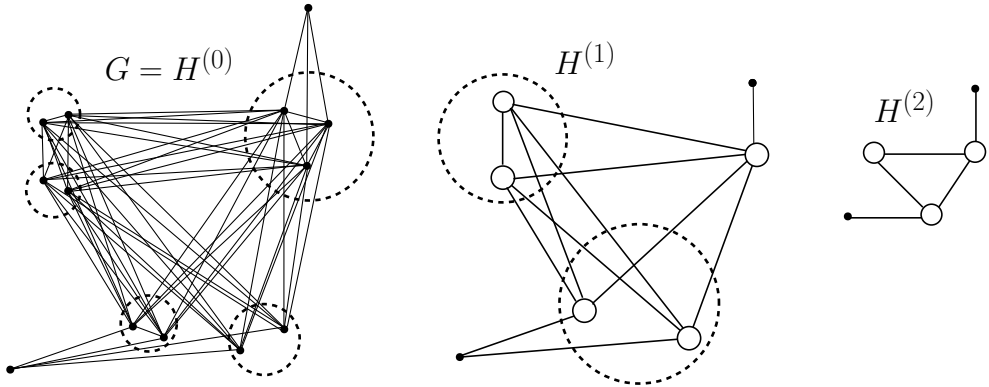}
		\caption{A  graph $G$ with iterated type partition $5$ and its corresponding type graph sequence: $G=H^{(0)},H^{(1)},H^{(2)}.$ Dashed circles group nodes having the same type. }
	\label{fig:img1}
\end{figure}
It is well-known that determining $nd(G)$ and the corresponding type
 partition, can be done in polynomial
  time \cite{L}. As an immediate consequence, we have that 
\begin{theorem}
{There exists a polynomial time algorithm, which 
 given a graph $G=(V,E)$, finds the type graphs sequence of $G$ and consequently the value $itp(G)$.}
\end{theorem}
\subsection{Relation with other parameters}

In this section we analyze the relations between  the   iterated type partition \ parameter and some  other well known parameters.

We  notice that, as an iteration  of  neighborhood diversity, the new parameter satisfies 
\begin{equation}\label{nd-itp}
itp(G)\leq nd(G).
\end{equation}
Actually   $itp(G)$ can be much smaller than  $nd(G).$ 
\remove{
Moreover, the family $\mathcal{F}_{\dnd}(k)$ of graphs  having  iterated type partition\ equal to $k$   is much larger than the family $\mathcal{F}_{nd}(k)$ of the  graphs  having neighborhood diversity equal to $k$.
Indeed a graph of the family $\mathcal{F}_{\dnd}(k)$ can be obtained by the following procedure:
\begin{itemize}
	\item choose a positive integer $d$ and a \rednote{connected} base graph  $H^{(d)}$ having $k$ nodes;
	\item For $i= d,d-1, \ldots, 1$
	\begin{itemize}
	\item	replace each  node of $H^{(i)}$  by a graph, which is an independent set of at least two nodes (when $i=d-j$ and $j$ is odd) or a clique of at least two nodes (when $i=d-j$ and $j$ is even). 
	\item for each edge of $H^{(i)}$, put a complete bipartite graph between the  nodes of the graphs that replace the endpoints of  the  edge;  
	\item rename the obtained graph as $H^{(i-1)}$
	\end{itemize}
\end{itemize}
The final graph $G=H^{(0)}$ belongs to  $\mathcal{F}_{\dnd}(k)$.
Notice  that given a graph $G$ from the family $\mathcal{F}_{\dnd}(k)$ the value of $nd(G)$ corresponds to the number of nodes in $H^{(1)}$ while the value of $itp(G)$ corresponds to the number $k$ of nodes in $H^{(d)}$.
 }
Indeed consider  the following: 
\begin{itemize}
	\item Choose a positive integer $d$ and  a {connected} base graph  $H^{(d)}$ having $k$ nodes;
	\item For $i= d,d-1, \ldots, 1$, a new  graph  $H^{(i-1)}$ is obtained as follows:
	\begin{itemize}
\item	{	replace each  node of $H^{(i)}$,  with an independent set of at least two nodes ({if $d-i$  is even}) or a clique of size at least two {(if $d-i$  is odd)}.  }  
	\item { for each edge of $H^{(i)}$, put a complete bipartite graph between the  nodes of the graphs that replace the endpoints of  the  edge.}
	\end{itemize}
	\end{itemize}
The value  $nd(H^{(0)})$ is the number of nodes in $H^{(1)}$, 
{that is at least $k 2^{d-1}$},
while  $itp(H^{(0)})$ is  the size $k$ of   $H^{(d)}$.

We stress that  iterated type partition \ 
is  a   ``special case'' of modular-width in which the 
modules in operation 4) can only be independent sets or cliques.
Hence,  it is not difficult to see that for every graph $G$ 
\begin{equation}\label{mw-itp} 
mw(G)\leq itp(G).
\end{equation}
We know from \cite{L} that $nd(G)\leq 2^{vc(G)}+vc(G)$. Hence, by (\ref{nd-itp}), we have $itp(G)\leq 2^{vc(G)}+vc(G)$.
Moreover, using the same arguments of \cite{L} is it possible to show that $cw(G)\leq itp(G)+1.$
Finally, as for the neighborhood diversity we can easily show that the  iterated type partition\ is incomparable to the treewidth by comparing the values of such parameters on a complete  graph $K_n$  and a path on $n$ nodes.
A summary of the relations holding between  some popular  parameters 
is given in Fig. \ref{relazioni}. {   We refer to \cite{Fomin2009} for the formal definitions of treewidth and  clique-width 
		parameters.}
\begin{figure}[tb!]
\centering
		\includegraphics[width=5truecm]{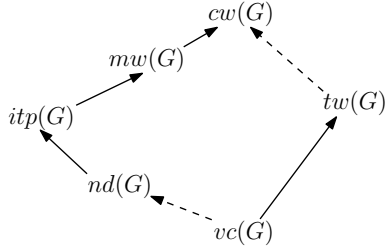} 
		\caption{A summary of the relations holding among  
		some popular  parameters. In addition to the previously  defined parameters,  
		we  use  $tw(G)$, $cw(G)$ and $vc(G)$ to denote  treewidth, clique-width and 
		minimum vertex cover of a graph $G$, respectively. Solid arrows denote 
		generalization, e.g., modular-width generalizes iterated type partition. 
		Dashed arrows denote that the generalization may exponentially increase the parameter. }
		\label{relazioni}
\end{figure}
\subsection{Our results and related work}
We give both tractability and hardness results for the new parameter.

\noindent
{\bf The \textit{Equitable Coloring} (EQC) problem.}
 If the nodes of a graph $G$  are colored with $k$ colors such that no adjacent nodes receive the same color (i.e., properly colored) and the sizes of any two color classes differ by at most one, then $G$  is said to be
{\em equitably $k$-colorable} and the coloring is said an 
{\em equitable $k$-coloring}.
The goal is to minimize the number of used colors.
The EQC problem is a well-studied problem, which has been analyzed in terms of parameterized  positive or negative results with respect to many different parameters \cite{Gomes}.

In particular, Fellows et al. \cite{Fel11} have shown that EQC problem  parameterized by treewidth and number of colors is $W[1]$-hard.
A series of reductions proving that \Eq \  is $W[1]$-hard for
different subclasses of chordal graphs are given in \cite{GLS}:
The problem is shown to be W[1]-hard  if parameterized by the
number of colors for block graphs and for the disjoint union of split graphs;  moreover,  
it remains W[1]-hard for $K_{1,4}$-free interval graphs even when parameterized by treewidth, 
number of colors and maximum degree.
In \cite{BF} an XP algorithm parameterized by treewidth is given.
We notice that an XP algorithm for \Eq \ parametrized by iterated type partition can be obtained by using Theorem 17 in \cite{Kn}.
On the other side, Fiala \textsl{et al.} show that 
the Equitable Coloring   problem is FPT when parameterized  
by the vertex cover number \cite{FGK}.
However, it was an open problem to establish the  parameterized  complexity of the Equitable Coloring problem parameterized  by   neigborhood diversity or
 modular-width. 
In section \ref{sec-eqc} we answer to these questions by proving the following results.
\begin{theorem}\label{teo-eqc}
The Equitable Coloring problem is $W[1]$-hard  parametrized by {\em itp}. 
\end{theorem}

Recalling (\ref{mw-itp}), Theorem \ref{teo-eqc} immediately gives that the Equitable Coloring Problem is $W[1]$-hard w.r.t. modular-width.
\begin{theorem}\label{cor-eqc}
The EQC problem is $W[1]$-hard  {parametrized by modular-width.}
\end{theorem}

We also show that \Eq\    $W[1]$-hardness drops  when parameterized by the neighborhood diversity.
\begin{theorem}\label{teo-eqc-nd}
The EQC problem is FPT when parameterized by  neighborhood diversity.
\end{theorem}

\smallskip
\noindent
{\bf FPT algorithms w.r.t. itp.}
In the last section 
we  deal with   FPT algorithms with respect to iterated type partition.
Some of the  considered problems   are already known to be FPT w.r.t  modular-width. Nonetheless, we   think that  
 the new algorithms, parameterized by   iterated type partition,  are  worthy to be considered, since they are much simpler, faster,  and allow to easily determine
not only the value, but also the optimal solution. {As an example we consider  here the  dominating set (DS),  the vertex coloring (Coloring), and the vertex cover (VC)  problems.
}

Table \ref{table} summarizes the contribution of this paper, in relation to known results. 

\begin{table}[tb!]
\begin{center}
\begin{tabular}{|l|l|l|l|}
\hline
  & DS, VC & Coloring  & EQC\\ 
\hline
 $cw$  & FPT\cite{Courcelle2000}	& W[1]-hard \cite{Fomin2009}&  W[1]-hard \cite{FGLS}\\   
 $mw$  & FPT\cite{romanek} & FPT\cite{GLO}   & W[1]-hard [*]\\     
 $itp$  & FPT(\com{2^t})[*]  & FPT(\com{t^{2.5t+o(t)}\log n})[*] &  W[1]-hard [*]\\  
 $nd$  & FPT\cite{L} & FPT\cite{L}  & FPT[*]\\    
 $vc$  & FPT\cite{L} & FPT\cite{L}   & FPT \cite{FGK} \\     \hline
\end{tabular}
\caption{The table summarizes the results known in literature for several problems parametrized by  iterated type partition\ and related parameters. $t$ denotes the value of the considered parameter and [*] denotes the result obtained in this paper.
\label{table}}
\end{center}
\end{table}
\section{Equitable coloring (EQC)}\label{sec-eqc} 
\def \l {{{\ell}}}
\def \P {{{P}}}
\def \Fao {{{F_{a,k}}}}
\def \Fa {{{F_{a}}}}
\def \Fu {{{F_{a_1}}}}
\def \Fj {{{F_{a_j}}}}
\def \Fh {{{F_{a_h}}}}
\def \FB {{{F_{B}}}}
\def \Fel {{{F_{a_\l}}}}
\def \Qko {{{Q_{k,\l,B}}}}
\def \Qk {{{Q}}}
\def \Quk {{{Q'}}}
\def \Qdk {{{Q''}}}
\def\Ho{{H^{(0)}}}
\def\H1{{H^{(1)}}}
\def\Hii{{H^{(i)}}}
\def\Hi1{{H^{(i-1)}}}
\def\Hd{{H^{(d)}}}
\def \Hdue {{H^{(2)}}}
\def \Htre {{H^{(3)}}}
\def \Hq {{H^{(4)}}}
 {In this section we 
prove 
Theorems \ref{teo-eqc} and \ref{teo-eqc-nd}.}
\begin{quote}
	\textbf{Equitable Coloring}
\\
\textbf{Instance:} A graph $G=(V,E)$ and an integer  $k$. 
\\
\textbf{Question:} 
Is it possible  to color the nodes of $G$ with exactly $k$ colors in such a way  that nodes connected by an edge receive different colors and each color class  has     either size $\lfloor |V|/k \rfloor$ or $\lceil |V|/k \rceil$? 
\end{quote}
\subsection{Hardness}
In order to prove that Equitable Coloring problem is  $W[1]$-hard if parameterized 
by iterated type partition,  we present a reduction  from the following 
Bin packing problem, which has been shown to be W[1]-hard when 
parameterized by the number of bins \cite{Jansen}.  
\begin{itemize}
	\item[] \textbf{\Bin}
\\
\textbf{Instance:} A collection of items $A=\{a_1,a_2,\cdots,a_{\l}\}$, a number $k$ of bins, and a bin capacity $B$.
\\
\textbf{Question:} $\exists$  a $k$-partition $\P_1,\cdots,\P_k$ of $A$ such that
  $\sum_{a_j \in \P_i} a_j = B$, $\forall \, i \in[k]$?
\end{itemize}
In general the \Bin \ problem asks for the sum of the items of each bin to be {\em at most} $B$; however, the above version   is equivalent to the general one (even from the parameterized point of view) as it is sufficient to add $kB-\sum_{j =1}^{\l} a_j$ unitary items \cite{GLS}. 
\begin{figure}[tb!]
	\centering
\includegraphics[width=0.7\textwidth]{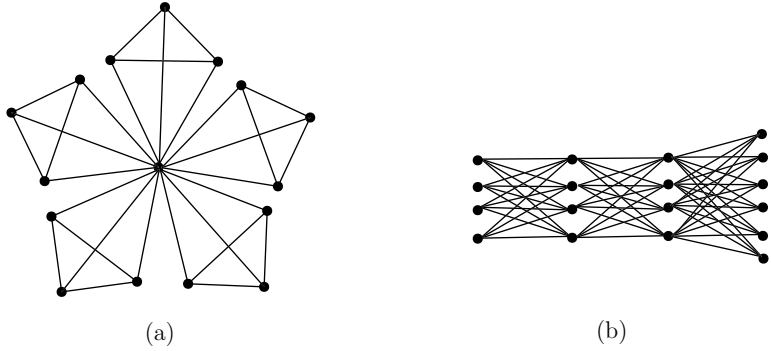} 
		\caption{(a) $(4,3)$--flower;  (b) $(3,5,4)$--chain. }
		\label{fig:flower+chain}
\end{figure}
In order to describe  our reduction,  we introduce two useful gadgets.
The first one is the flower gadget also used in \cite{GLS}. 
Let $a$ and $k$ be positive integers.
An  {\em $(a,k)$--flower} $\Fao$  is a graph obtained 
by joining $a +1$ cliques  of size $k$ to a central node $y$. 
Fig. \ref{fig:flower+chain}(a) shows the $(4,3)$--flower.
Formally, {let $K_k^i$ be a copy of a cliques  of size $k$, for each $i \in[a+1]$, }

-- \ $V(\Fao)=\{y\} \cup \bigcup_{i \in[a+1]}V(K_k^i)$, and 

-- \ $E(\Fao)= \{(y,x) \ | \ x \in \bigcup_{i \in[a+1]}V(K_k^i)\} 
\cup \bigcup_{i \in[a+1]}E(K_k^i).$\\
\noindent
The second gadget is defined starting from three positive integers: 
$k, \l$ and $B$. It is a sequence of sets of independent nodes 
 $S_1, \cdots, S_k, S_{k+1}$ with $|S_i|=B$, for $i \in [k]$, 
and $|S_{k+1}|=\l+1$ where between each pair of consecutive sets in 
the sequence $S_i$, $S_{i+1}$ there is a complete bipartite graph. 
We call such a gadget a $(k,\l,B)$--{\em chain} $\Qk$. 
Fig. \ref{fig:flower+chain}(b)  shows the $(3,5,4)$--chain.
Formally, 

-- \ $V(\Qk)= \bigcup_{i \in[k+1]}S_i$, and

-- \ $E(\Qk)= \bigcup_{i \in[k]}\{(u,v) \ | \ u \in S_i, v \in S_{i+1}\}.$

We can now  describe our reduction. 
Let $\langle A=\{a_1,\cdots,a_\l\}, k, B \rangle$ be an instance of \Bin.
Define a graph $G$ as follows: 
{Consider the disjoint union of two $(k,\l,B)$-chains, $\Quk$ and $\Qdk$, and the flowers $F_{a_1,k}, \cdots, F_{a_\l,k}, F_{B,k}$, then join each node in the flowers to each node in the chains. In the following, whenever the number of bin $k$ is clear by the context, we use $\Fa$ instead of $\Fao$.}
Formally,

-- $V(G)= V(\Quk) \cup V(\Qdk) \cup V(\FB) \cup \left(\bigcup_{j \in [\l]} V(\Fj)\right)$, and 

-- $E(G) = E(\Quk) \cup E(\Qdk)  \cup E(\FB) \cup \left(\bigcup_{j \in [\l]} E(\Fj) \right)\cup 
\\
\hphantom{E(G) oooooooooooooooooo==} \cup \left\{(x,u) \big| \ x \in V(\FB) \cup \left(\bigcup_{j \in [\l]} V(\Fj)\right), \ 
u \in V(\Quk) \cup V(\Qdk)  \right\}.$

\noindent
Fig. \ref{fig:type-graph-sequence}  shows the graph $G$ when 
{$A=\{2,1,2,3\}$}, $B=4$ and $k=3$. {Call $S'_i$ (resp. $S''_i$) is the $i$-th set of 
independent nodes in $\Quk$ (resp. $\Qdk$).}
The number of nodes in the resulting graph $G$ is
\begin{eqnarray} \label{size_G} \nonumber
|V(G)| & = & |V(\Quk)|+ |V(\Qdk)| + |V(\FB)| + \sum_{j \in [\l]} |V(\Fj)|  \\
\nonumber
 & = & \sum_{i \in[k+1]} |S'_i| +  \sum_{i \in[k+1]} |S''_i| +
(1+ (B+1)k) + \sum_{j \in [\l]} (1+(a_j+1)k) \\
\nonumber
 & = & 2(Bk +\l+1) + (1+ (B+1)k) + \l + k  \sum_{j \in [\l]} (a_j+1) \\
\nonumber
 & = & 2(Bk +\l+1) + (1+ (B+1)k) + \l + k \l +k^2B\\
 & = & (k+3)(Bk+\l +1).
\end{eqnarray}

\begin{figure}[tb!]
	\centering
			\includegraphics[width=14truecm]{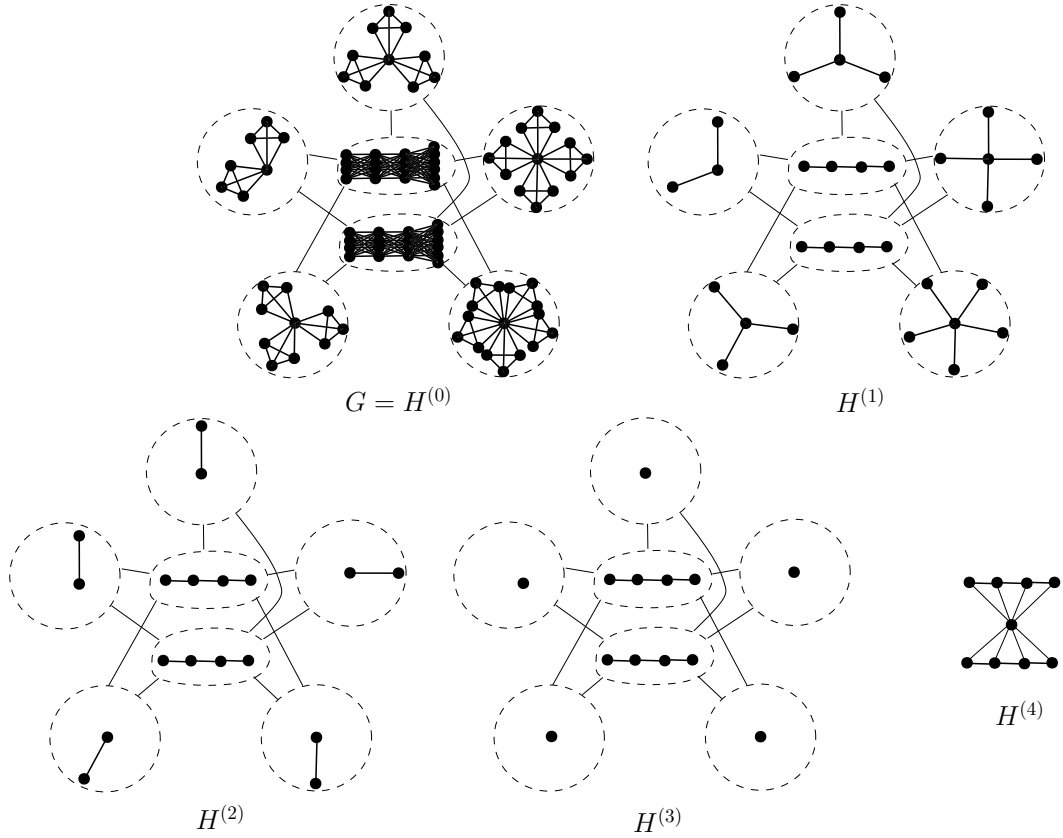} 
		\caption{The type graph sequence 
of $G$ when {$A=\{2,1,2,3\}$}, $B=4$, and $k=3$. The line connecting dashed circles indicates a complete bipartite graph between the nodes in the circles. }
		\label{fig:type-graph-sequence}
\end{figure}

\begin{lemma}\label{lemma-iff}
$\langle A=\{a_1,\cdots,a_\l\}, k, B \rangle$ is a YES instance of \Bin \ 
if and only if $G$ is equitably $(k+3)$--colorable. 
\end{lemma}

\begin{proof}{}
Given a $k$-partition $\P_1,\cdots,\P_k$ of $A$ that solves our instance 
of \Bin, i.e.,  $\sum_{a_j \in \P_i} a_j = B$ for each $i \in [k]$, 
we  construct a coloring $c$ of the nodes of $G$ and prove 
that it is an equitable (k+3)-coloring of $G$. 
\begin{itemize}
\item{}{ Coloring of the nodes in $\Quk$:} 
For each $i \in [k+1]$ and  $u \in S'_i$ , assign
\begin{equation} \label{S11}
c(u)= \begin{cases} 
         k+3 & \mbox{if $i$ is odd,}\\
				 k+2 & \mbox{if $i$ is even.}
       \end{cases}
\end{equation}
\item{}{Coloring of the nodes in $\Qdk$: }
For each $i \in [k+1]$ and $u \in S''_i$,   assign
\begin{equation} \label{S22}
c(u)= \begin{cases} 
         k+3 & \mbox{if $i$ is even,}\\
				 k+2 & \mbox{if $i$ is odd.}
       \end{cases}
\end{equation}
\item{}{ Coloring of the nodes in $\FB$: }
Let $z$ be the central node in 
$\FB$. Assign $c(z)=k+1$. Then, assign to each of the $k$ nodes of 
the $B+1$ cliques joined to $z$ the remaining $k$ colors, i.e., 
the colors in $\{1,2, \cdots k\}$, so that each node of the clique has a different color. 
\item{ Coloring of the nodes in $\Fj$, for $j\in[\l]$:} 
Let $y_j$ be the central node in 
$\Fj$. Assign $c(y_j)=i$ if $a_j \in P_i$. Then, as before assign to each of 
the $k$ nodes of the $a_j+1$ cliques joined to $y_j$ the remaining $k$ 
colors, i.e., the colors in $\{1,2, \cdots k,k+1\}-\{i\}$, so that each node of the clique has a different color. 
\end{itemize}
It is immediate to see that the above coloring $c$ is   proper.
Now we prove that is also equitable. Since  $|V(G)| = (Bk +\l+1) (k+3)$ 
(recall (\ref{size_G})), we have only to prove that each class of 
colors contains $Bk +\l+1$ nodes.
Denote by $C_i$ the class of color $i$, with $i\in [k+3]$.
\\
$\bullet$ Colors $k+3$ and $k+2$ are used only to color the nodes of the two $k$-chains $\Quk$ and $\Qdk$ and by (\ref{S11}) and (\ref{S22}) we have\\
$$|C_{k+3}|= \begin{cases} 
    |S'_1|+|S''_2|+\cdots + |S'_k|+|S''_{k+1}|= kB+\l+1 & \mbox{ if $k$ is odd}\\
    |S'_1|+|S''_2|+\cdots + |S''_k|+|S'_{k+1}|= kB+\l+1 & \mbox{ if $k$ is even.}
	\end{cases}$$	
$$|C_{k+2}|=  \begin{cases}
|S''_1|+|S'_2|+\cdots + |S''_k|+|S'_{k+1}|= kB+\l+1 & \mbox{ if $k$ is odd} \\
 |S''_1|+|S'_2|+\cdots + |S'_k|+|S''_{k+1}|= kB+\l+1 & \mbox{  if $k$ is even.}
	\end{cases}$$	
$\bullet$ Color $k+1$ is used to color node $z$ and exactly one node in each of the $a_j+1$ cliques in the flower $\Fj$ for $j \in [\l]$; hence,
$$|C_{k+1}|= 1 + \sum_{j \in [\l]} (a_j+1) = 1 + \l + kB.$$
$\bullet$ Let $i \in [k]$. Colors $i$  is used to color the following nodes:
The central node $y_j$ of the flower $\Fj$ where $a_j \in P_i$ 
(no other node in  $\Fj$  is colored with $i$), exactly one node in each
 of the $a_h+1$  cliques in the flower $\Fh$ for $a_h \not\in P_i$,
and exactly one node in each of the $B+1$ cliques in the flower $\FB$. Hence,

\begin{eqnarray*}
|C_{i}| &=& \sum_{a_j \in P_i} 1 + \sum_{a_h \not\in P_i} (a_h+1) + (B+1)\\
&=& \sum_{a_j \in A} 1  + \sum_{a_h \not\in P_i} a_h + (B+1)\\
&=& \l + (k-1)B +(B+1)= 1 + \l + kB.
\end{eqnarray*}

\smallskip

Now, let $c$ be an equitable $(k+3)$-coloring of $G$.
First we claim some features of the coloring $c$ and then we use them to 
build a $k$-partition of $A$.
\begin{claim} \label{claim1}
The coloring $c$ assigns two colors to the nodes in $\Quk$ and $\Qdk$, and these colors are not assigned to any node of the flowers $\FB$ and 
$\Fj$ for $j\in [\l]$.
\end{claim}
\proof 
Since the central node $y_j$ (resp. $z$) of the flower $\Fj$ (resp. $\FB$) 
is connected to each node in any clique $K_k$ in $\Fj$ (resp. $\FB$), we have that the nodes of $\Fj$ (resp. $\FB$) need at least $k+1$ colors. 
Furthermore, the nodes of $\Fj$ and of $\FB$ are connected to all the nodes 
of the $k$-chains $\Quk$ and $\Qdk$. Hence, the colors that $c$ assigns to the nodes in $\Quk$ and $\Qdk$ have to be at most two (recall that $c$ is a 
$(k+3)$-coloring). 
On the other hand, between each pair of consecutive sets $S'_i$, $S'_{i+1}$ in $\Quk$  (resp. $S''_i$, $S''_{i+1}$ in $\Qdk$) there is a complete bipartite graph; hence, the coloring $c$ has to assign to the nodes in $\Quk$ and $\Qdk$ at least two colors.
\qed 

By Claim \ref{claim1}, w.l.o.g. we assume that $c$ assigns colors $k+2$ and $k+3$ to the nodes in $\Quk$ and $\Qdk$, and colors in $[k+1]$ to the nodes of the flowers $\FB$ and $\Fj$ for $j \in [\l]$.
In the following we denote by $C_i$ be the class of nodes whose color is $i$,
where $i \in [k+1]$.

\begin{claim} \label{claim2}
$c(z) \neq c(y_j)$ for each $j \in [\l]$.
\end{claim}
\proof 
Let $c(z) = i$ with $i \in [k+1]$.
By contradiction assume that there exists   a node $y_j$ 
for some $j \in [\l]$, such that $c(z) = c(y_j) = i$, that is $y_j \in C_i$.
Since $z$ is connected to each other node in the flower 
$\FB$ and $y_j$ is connected to each other node in the flower $\Fj$, we have that color $i$ is not used by any other node in  $\FB$ and $\Fj$. 
On the other hand, color $i$ is used by
exactly one node in each of the $a_h+1$ cliques in flower $\Fh$ where 
$y_h \not \in C_ i$ 
(recall that $c$ assigns colors in $[k+1]$ to the nodes of $\Fh$).
Hence,
\begin{eqnarray*}
|C_i| & = & 2 + \sum_{h: y_h \not \in C_i}(a_h+1) 
= 2 + \sum_{ h: y_h \not \in C_i}(a_h+1) + \sum_{j: y_j \in C_i}(a_j+1) -  \sum_{j: y_j \in C_i}(a_j+1)\\
& = & 2 + \sum_{h \in [\l]}(a_h+1) -  \sum_{j: y_j \in C_i}(a_j+1)\\
&= &2+ \l + kB - \sum_{j: y_j \in C_i}(a_j+1) \\
&\leq  & \l + kB  \qquad  \quad \mbox{(by the hypothesis there exists $y_j \in C_i$  and  $\sum_{j: y_j \in C_i}(a_j+1)\geq 2$)}.
\end{eqnarray*}
The above inequality is not possible since $c$ is an equitable $(k+3)$--coloring 
and by (\ref{size_G}) we have $|C_i|=1+ \l + kB$.
\qed

By Claim \ref{claim2}, w.l.o.g. we assume that $c(z)=k+1$ and 
then $c(y_j) \in [k]$ for each $j \in [\l]$.
In the following we will prove that the partition 
$P_i=\{a_j \ | \ c(y_j)=i \}$ with $i \in [k]$ is a $k$-partition of $A$.
In particular, we will prove that 
\begin{equation}\label{partition}
\sum_{j: y_j \in C_i} a_j =B.
\end{equation}
Consider a flower $\Fj$, with $j \in [\l]$:
 If the color $i$ is assigned to the center $y_j$  then it is not assigned to any other vertex in $\Fj$;
 if, otherwise, the color $i$ is not assigned to the center $y_j$ then it is  assigned to exactly one vertex in each of the $a_j+1$ cliques $K_k$ connected to $y_j$.\\
Furthermore, the center $z$ of the flower $\FB$  has color $k+1$; hence,  color $i$  is  assigned to exactly one vertex in each of the $B+1$ cliques $K_k$ connected to $z$.
Summarizing,
\begin{eqnarray} \nonumber
|C_i| & = & \sum_{j: y_j \in C_i} 1 + (B+1) + 
 \sum_{h: y_h \not \in C_i} (a_h+1)\\  \nonumber
& = & \sum_{j: y_j  \in C_i} (1 +a_j-a_j) +  (B+1) +
 \sum_{h: y_h \not \in C_i} (a_h+1)\\  \nonumber
& = &  \sum_{j \in [\l]} (1 +a_j) +  (B+1) - 
 \sum_{j: y_j  \in C_i} a_j\\ \label{Ci}
& = &  kB + \l  + (B+1) - 
 \sum_{j: y_j  \in C_i} a_j
\end{eqnarray}
Since the coloring $c$ is an equitable (k+3)-coloring and by (\ref{size_G}) 
it holds $|V(G)| = (Bk +\l+1) (k+3)$, we have that  
$|C_i|=Bk +\l+1$. By using this fact and (\ref{Ci}) we have (\ref{partition}).
\end{proof}

\begin{lemma}\label{lemma-itp}
The iterated type partition $\dnd(G)$ of $G$ is $2k+3$. 
\end{lemma}
\begin{proof}{}
\remove{
We will show the type graph sequence $\Ho=G, \H1, \Hdue, \Htre, \Hq$ 
of $G$, and we will verify that the number of nodes of the base 
graph $\Hq$ is $2k+3$.
Fig. \ref{fig:type-graph-sequence}  shows the type graph sequence 
of graph $G$ when $A=\{1,2,3,4,2\}$, $B=4$ and $k=3$.
}

\noindent
The type graph $\H1$ of $G$ is obtained as follows:\\
-- Compress the  cliques   in the flower $\FB$ into one node each. 
Call  $f_{01}, \ldots, f_{0(B+1)}$ the resulting nodes.\\
-- Compress each of the $a_j+1$ cliques $K_k$ in the flower $\Fj$, for $j \in [\l]$, into one node. 
Call them $f_{j1}, \ldots, f_{j(a_j+1)}$.\\
-- Compress each set $S'_i$ of independent nodes in $\Quk$ into one node.
Call them $s'_1, \cdots s'_{k+1}$.\\
-- Compress each set $S''_i$ of independent nodes in $\Qdk$ into one node.
Call them $s''_1, \cdots s''_{k+1}$.\\
As a consequence we have: \\
$V(\H1) = \{z, f_{01}, \ldots, f_{0(B+1)} \}\cup
          \bigcup_{j \in[\l]} \{y_j, f_{j1}, \ldots, f_{j(a_j+1)} \} \cup
					\{ s'_1, \cdots s'_{k+1} \}  \cup \{ s''_1, \cdots s''_{k+1} \}$
\begin{eqnarray*}
E(\H1) & = & \{(z, f_{0i}) \ | \ i \in[B+1] \}\cup
          \bigcup_{j \in[\l]} \{(y_j, f_{ji}) \ | \  i \in[a_j+1] \} \cup
	\{ (s'_i, s'_{i+1}),  (s''_i, s''_{i+1}) \ | \  i \in[k] \}  \cup \\
	& &				\{(x,v) \ | \ x \in\{z, f_{0i} \ | \ i \in [B+1] \} \cup
					          \{y_j, f_{ji}  \ | \ j \in[\l], i \in [a_j+1] \},
										v \in \{s'_i, s''_i \ | \ i \in [k+1]\} \}
\end{eqnarray*}

\noindent
The type graph $\Hdue$ of $\H1$ is obtained as follows:\\
-- Compress the set of independent nodes $\{f_{01}, \ldots, f_{0(B+1)}\}$ into a node. Call it 	$f_0$. \\
-- Compress the set of independent nodes $\{f_{j1}, \ldots, f_{j(a_j+1)}\}$ into a node. Call it 	$f_j$.\\
Hence, 
$V(\Hdue) = \{z, f_0\} \cup  \bigcup_{j \in[\l]} \{y_j, f_j\} \cup 
            \{ s'_1, \cdots s'_{k+1} \}  \cup \{ s''_1, \cdots s''_{k+1} \}$
and
\begin{eqnarray*}
E(\Hdue) & = & \{(z, f_0)\}  \cup  \bigcup_{j \in[\l]} \{(y_j, f_j)\} 
            \cup
					\{ (s'_i, s'_{i+1}), (s''_i, s''_{i+1})\ | \  i \in[k] \} \cup \\
		& & \{(x,v) \ | \ x \in\{z, f_0 \}\cup \{y_j, f_j  \ | \ j \in[\l] \},
										v \in \{s'_i, s''_i \ | \ i \in [k+1]\} \}
\end{eqnarray*}

\noindent
The type graph $\Htre$ of $\Hdue$ is obtained as follows:\\
-- Compress the clique consisting of one edge $\{(z, f_0)\}$ into a node. Call it $z'$.\\
-- Compress the cliques each consisting of one edge $\{(y_j, f_j)\}$ into a node, for $j \in [\l]$. Call it $y'_j$.\\
Hence, 
$V(\Htre) = \{z', y'_j \ | j \in[\l]\} \cup 
            \{ s'_1, \cdots s'_{k+1} \}  \cup \{ s''_1, \cdots s''_{k+1} \}$
and
\begin{eqnarray*}
E(\Htre) & = & 
					\{ (s'_i, s'_{i+1}), (s''_i, s''_{i+1})\ | \  i \in[k] \}\cup \\
		& & \{(x,v) \ | \ x \in\{z',  y'_j \ | j \in[\l]\},
										v \in \{s'_i, s''_i \ | \ i \in [k+1]\} \}
\end{eqnarray*}

\noindent
The type graph $\Hq$ of $\Htre$ is obtained as follows:\\
-- Compress the set of independent nodes  $\{z', y'_j \ | j \in[\l]\}$ into a node. Call it $y''$.\\
Hence, 
$V(\Hq) = \{y''\} \cup 
            \{ s'_1, \cdots s'_{k+1} \}  \cup \{ s''_1, \cdots s''_{k+1} \}$
and
\begin{equation*}
E(\Hq) = 
					\{ (s'_i, s'_{i+1}), (s''_i, s''_{i+1})\ | \  i \in[k] \}\cup  \{(y'',v) \ | 	v \in \{s'_i, s''_i \ | \ i \in [k+1]\} \}
\end{equation*}
It is immediate to see that $\Hq$ is a base graph and that $|V(\Hq)|=2k+3$.
\end{proof}

\begin{proof}{ \emph{of Theorem \ref{teo-eqc}.}}
Given an instance $\langle A=\{a_1,\cdots,a_\l\}, k, B \rangle$ of \Bin,
we use the above construction to create an instance $\langle G=(V,E), \dnd(G) \rangle$ of \Eq parameterized by iterated type partition.
Lemma \ref{lemma-iff} show the correctness of our reduction and Lemma \ref{lemma-itp} provides the iterated type partition of the constructed graph, showing that our new parameter $\dnd(G)$ is linear in the original parameter $k$. 
\end{proof}

\subsection{Neighborhood Diversity: an FPT algorithm}
We prove here that the Equitable Coloring problem admits a FPT algorithm  with respect to neighborhood diversity.
W.l.o.g. we assume that the number  of nodes in the  input graph $G=(V,E)$ is a multiple of the number of colors $k$ (this can be  attained by adding a clique of $|V|-\left(\left\lfloor |V|/ k \right \rfloor \cdot k \right)$ nodes connected to a node in $G$ in such a way the answer to the equitable $k$-coloring question remains unchanged).
\\
Let then $r=|V|/k$. Any equitable $k$-coloring  of $G$ partitions $V$ into $k$ classes of colors, say $C_1,\ldots,C_k$, s.t.  $C_\ell$ is an independent set of G of size $|C_\ell|=r$, for $\ell=1,\ldots, k$.

If we consider now  the type partition $\{V_1, \ldots, V_\t\}$ of $G$ and the corresponding   type graph $H=(V(H)=\{1,\ldots,t\}, E(H))$, we trivially have that: {\em Two nodes $u,v\in V$ are independent in $G$  iff $v\in V_i$ and $u\in V_j$, with $i,j\in V(H)$,  such that either $i$ and $j$ are independent nodes of   $H$ or $i=j$ and $V_i$ induces an independent set in $G$.}
This immediately implies that for each color class $C_\ell$ of the equitable coloring of $G$ there exists an independent set   $I_\ell=\{{\ell_1},\ldots,{\ell_{\rho}}\}$ of $H$ such that

 $\sum_{s=1}^{\rho} |C_\ell\cap V_{\ell_s}|=r$ and  

 $|C_\ell\cap V_{\ell_s}|=1$  for each $s=1,\ldots,\rho$   such that  $V_{\ell_s}$ induces a clique.

\noindent
Let now $\I$ denote the family of all independent sets in $H$. From the above reasoning we have that, given any equitable $k$-coloring  of $G$, we can  associate  to each  $I\in \I $  a separate set of $z_I\geq 0$ colors so that 
\begin{enumerate}
\item   $\sum_{I\in \I} z_I=k$, 
\item  for each $i\in V(H)$ it holds 
that the sum over all $I\in \I$ such that $i\in I$ of the number of nodes in $V_i$ that (in the coloring of $G$) are colored with one of the $z_I$ colors associated to $I$ (this number is at most $z_I$ if $V_i$ induces a clique in $G$, but can be larger if $V_i$ induces  an independent set) is exactly $|V_i|$.
\item  for each $I\in \I$  it holds 
that the sum over all $i\in V(H)$ of the number of nodes  of $V_i$ that that are colored in $G$  with one  of the $z_I$ colors associated to $I$ is  $r\cdot z_I$.
\end{enumerate}
The above conditions can be expressed by the following linear program on the
 variables  $z_I$ for each $I\in \I$ and  $z_{I,i}$  for each $I\in \I$ and for each $i\in I$.
\begin{enumerate}
 \item   $\sum_{I\in \I} z_I=k$;
   \item   $\sum_{I\, :\, i\in I} z_{I,i}=|V_i|$,  \mbox{for each $i\in V(H)$};
  \item    $\sum_{i\in I} z_{I,i} -r\cdot z_I=0$,  \mbox{for each $I\in \I$};
  \item    $z_{I} - z_{I,i}\geq 0$  for each $I\in \I$ and $i\in I$ such that $V_i$ is a clique;
  \item    $z_{I,i} \geq 0$  for each $I\in \I$ and $i\in V(H)$.
\end{enumerate}
From the above reasoning, it is clear that if the graph $G$ admits an equitable $k$-coloring, then there exists an assignation of values to the variables 
$z_I$  and  $z_{I,i}$,   for each $I\in \I$ and  $i\in I$, that satisfies the above system.

We show now that from any  assignation of values to the variables 
$z_I$  and  $z_{I,i}$ that satisfies the above system, we can obtain an equitable $k$-coloring  of $G$.

\smallskip
\noindent
 $\bullet$ For each independent set $I\in \I$, such that $z_I>0$, repeat the following procedure:
\begin{itemize}
\item Select a   set of $z_I$ new colors, say $c^I_1, \dots, c^I_{z_I}$ (to be used only for nodes in $I$);\\
 We notice that (by 3.) the total number of nodes to be colored is $r\cdot z_I$;
\item Consider the list of colors     $c^I_1,c^I_2, \dots, c^I_{z_I}, c^I_1,c^I_2, \dots, c^I_{z_I}\ldots, c^I_1,c^I_2, \dots, c^I_{z_I}$ (obtained cycling  for $r$ times on $c^I_1,\dots, c^I_{z_I}$); assign the colors starting from the beginning of the list as follows:
For each  $i\in V(H)$,    select  $z_{I,i}$ uncolored nodes in $V_i$ (it can be done by 2.)  and assign to them the next unassigned $z_{I,i}$ colors in the  list. 
\end{itemize}
  In this way  each color is used exactly $r$ times. Moreover, since each independent set uses a separate set of colors,   the total number of colors  is $\sum_{I\in \I} z_I=k$ (crf. 1.).
 Furthermore,  in  each $V_i$ that induces a clique in $G$,  we color   $z_{I,i}\leq z_{I}$  nodes (this holds by 4.). Such nodes  get  colors which are consecutive  in the list, hence they are different. 
Summarizing, the desired equitable $k$-coloring  of $G$   has been obtained.

Finally, we  evaluate the time to solve the above system.
We use the well-known result that Integer Linear Programming is FPT parameterized by the number of variables.
\begin{itemize}
\item[] \textbf{$\ell$-Variable Integer Linear Programming Feasibility}
\\
\textbf{Instance:}  A matrix $A \in Z^{m\times \ell}$ and a vector $b \in Z^m$.
\\
\textbf{Question:} Is there a vector $x \in Z^\ell$  such that $Ax \geq b$?
\end{itemize}
\begin{proposition} \label{prop}\cite{FLMRS}
$\ell$-Variable Integer Linear Programming Feasibility 
can be solved in time $O(\ell^{2.5t+o(\ell)} \cdot L)$ where $L$ is the number of bits in the input.
\end{proposition}
{ Since $|V(H)|=\nd(G)$, our system uses at most $O(\nd(G) 2^{\nd(G)})$ variables: $z_I$  for  $I\in \I$  and  $z_{I,i}$ for  $I\in \I$ and  $i\in I$.
We have $O(\nd(G) 2^{\nd(G)})$  constraints and the coefficients are upper bounded by $r=|V|/k$. Therefore, Theorem \ref{teo-eqc-nd} holds.}

\section{Algorithms}\label{sec-algo}

In this section, we provide some FPT algorithms with respect to iterated type partition.
 In order to solve a problem $P$ on an input graph $G$,   the general algorithm scheme  is: 
\begin{enumerate}
	\item[1)] Iterate by generating the whole type graph  sequence of $G$. 
	\item[2)] On each graph $G'$ in the type graph sequence, a  generalized version $P'$  of the original problem is defined (with $P'$ in $G'$ being equivalent to $P$ in $G$). 
	\item[3)] Optimally solve $P'$  on the base graph and reconstruct the solution on the reverse type graph sequence (hence solving $P$ in $G$).
\end{enumerate}
If the    construction of  the solution for  $P'$ (at  step 2), can be done  in polynomial time and the time to solve $P'$ on the base graph is {$f$,} 
then the whole algorithm needs  {\com{f}} time.
We stress that this is indeed the case for the  algorithms below.

\def \stds{{stds}}

\subsection{Dominating set}
In order to present our FPT algorithm for the minimum dominating set problem in $G$ with parameter $\dnd(G)$, we consider  the following generalized  dominating set problem.
\begin{definition}
Given a graph $G=(V,E)$ and a set of nodes $Q \subseteq V$, 
a  {\em semi-total dominating set} of $G$  with respect to $Q$, 
called $Q$-{\em \stds} of $G$, is a set $D\subseteq V$ such that every node in $Q$ is adjacent to a node in $D$, and 
every other node  is either a node in $D$ or it is dominated by 
a node in $D$. 
  The set $D$ is said an {\em optimal} $Q$-{\em \stds} of $G$, if its size is minimum among all the $Q$-{\em \stds} of $G$.
\end{definition}
Clearly, when $Q=V$ the semi-total dominating set problem  is the  total domination problem 
\cite{CDH}.
If $Q=\emptyset$, the semi-total dominating set problem 
becomes the dominating set  problem.

\begin{lemma}\label{only-one}
Let $G=(V,E)$ be a connected graph and 
let $\V=\{V_1,\cdots, V_t\}$ be the type partition of $G$. 
Let $Q \subseteq V$.
There exists an optimal $Q$-\stds \ $D$ of $G$ such that 
\begin{equation}\label{one}
|V_x \cap D| \leq 1  \qquad \mbox{for each $x\in[t]$. }
\end{equation}
\end{lemma}
\proof

Let $D$ be an optimal $Q$-\stds \  of $G$.
Assume there exists  $x\in[t]$  such that $|V_x \cap D| \geq 2$.
We distinguish two cases according to  $V_x$ being  a clique or an independent set.

Let $V_x$ be a clique. Let $u$ and $v$ be two nodes in $V_x \cap D$. 
Let $u \not \in Q$.
Since $u$ is a neighbor of $v$ and since $u$ and $v$ share the same neighborhood, we have that the set $D'=D-\{v\}$ is a $Q$-\stds \  of $G$. Furthermore, $|D'| < |D|$ and this is not possible since $D$ is optimal.
Assume now that $u \in Q$. If there exists a neighbor $w$ of $u$ with 
$w \in V_y \cap D$, for some $y\neq x$, then as above $D'=D-\{v\}$ is a $Q$-\stds \  of $G$ and  $|D'| < |D|$.
If, otherwise, node $u$ has no neighbor in $D$ except for those in $V_x$, then we can choose any neighbor $w$ of $u$ with 
$w \in V_y \cap D$,  for some $y\neq x$, and $D'=D-\{v\} \cup\{w\}$ is a $Q$-\stds \  of $G$ and  $|D'| = |D|$.

Let $V_x$ be an independent set. Let $u$ be any node in $V_x \cap D$.
If there exists a neighbor $w$ of $u$ with $w \in V_y \cap D$, 
for some $y\neq x$, then the set $D'$ obtained from $D$ removing all 
the nodes in $V_x$ except for $u$ is again a $Q$-\stds\  since the neighbors of nodes in $V_x$ are dominated by $u$ and all the nodes in $V_x$ are dominated by $w\in V_y$. Furthermore, $|D'| < |D|$.
Otherwise, we have that $V_x \subset D$ and for each neighbor 
$w$ of $u$ it holds $w \in V_y$, for  $y\neq x$, and  $w \not \in D$.
Hence, the set $D'$ obtained from $D$ removing all the nodes in $V_x$ 
except for $u$ and adding to it a node $w \in V_y$, where $y$ is such that 
$V_y\cap D=\emptyset$, is a $Q$-\stds \  of $G$. 
Furthermore, $|D'| \leq |D|$.

Repeating the above argument for each $x\in[t]$  such that $|V_x \cap D| \geq 2$,
we obtain an optimal solution satisfying (\ref{one}).
\qed

The FPT algorithm \Domination$\ $ recursively constructs the graphs in the type graph sequence of $G$, until the base graph is obtained.
It is initially called with \Domination($G,\emptyset$).
At each recursive step, the algorithm \Domination($H ,Q $), on  a graph $H$ and a set 
$Q\subseteq V(H)$ of nodes that need to have a  neighbor in the solution set,
checks if $H$ is a base graph or not.
In case $H$ is a base graph, then
the algorithm searches by brute force the $Q$-stds of $H$ 
and returns it.
If $H$ is not a base graph  then the algorithm 
first constructs the type graph $H'$
and conveniently selects nodes in $V(H')$ to assemble a set 
$Q'$ of nodes that need to have a  neighbor in the solution set, then it uses 
the set $D'$ of nodes in $V(H')$ returned by \Domination($H',Q'$)
to construct the output set $D \subseteq V(H)$. The nodes of the returned set $D$ are chosen 
selecting exactly one node from each metavertex $V_x$ having $x \in D'$.

Figure \ref{fig:alg1} gives an example of the execution of  Algorithm \ref{alg} on the graph $G$ in  Fig. \ref{fig:img1}.

\begin{algorithm}[tb!]
\SetCommentSty{footnotesize}
\SetKwInput{KwData}{Input}
\SetKwInput{KwResult}{Output}
\DontPrintSemicolon
\caption{ \ \   \textbf{Algorithm} \Domination($H,Q$) \label{alg}}
\KwData { A graph $H=(V(H),E(H))$, a set $Q\subseteq V(H)$.
}
\setcounter{AlgoLine}{0}
\uIf{ $H$ is a base graph }{ 
$D=V(H)$\\
\lFor{\textbf{each} $S\subseteq V(H)$}{ 
	\textbf{if} {(($S$ is $Q$-\stds \ of $H$) \textbf{and} ($|S|<|D|$)) \textbf{then} $D=S$}}
}
\uElse{ 
Let $V_1,\cdots, V_t$ be the type partition of $H$ and let $H'$ be the type graph of $H$. \\
$Q'=\{ x \in V(H') \ | \  (V_x \cap Q\neq\emptyset   \mbox{ or } 
 V_x \mbox{ is an independent set})\}$\\
$D'$ = \Domination$(H',Q')$ \\
$D=\bigcup_{x \in D'} \{u_x\}$, where $u_x$ is an arbitrarily chosen node in $V_x$\\
}
 \Return $D$
\end{algorithm}
\begin{figure}[tb!]
	\centering
	\includegraphics[width=\textwidth]{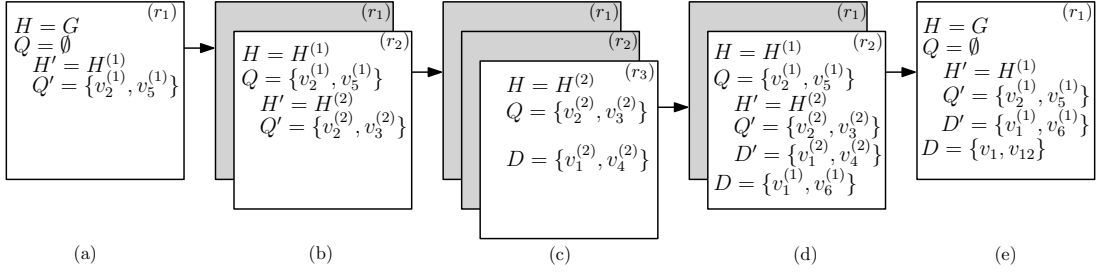} \vspace*{-0.6truecm}
		\caption{The recursive execution of the algorithm \ref{alg} on the graph $G$ depicted in Fig. \ref{fig:img1}: ((a) and (b)), since the input graph is not a base graph, their type partition as well as the set $Q'$ are computed and passed to the next recursive level; (c), $H$ is a base graph and then an optimal solution is computed exploiting a brute force approach;  ((d) and (e)),  an optimal solution $D=\{v_1,v_{12}\}$ is reconstructed using the solution $D'$ obtained on the reverse type
graph sequence.}
	\label{fig:alg1}
\end{figure}
\begin{lemma} \label{recursive-step}
Let $H$ be not a base graph and let $Q \subseteq V(H)$.
Let $V_1,\cdots, V_t$ be the type partition of $H$ and let $H'$ be its type graph.
If $Q'=\{ x \in V(H') \ | \  V_x \cap Q\neq\emptyset  \mbox{ or } 
 V_x \mbox{ is an independent set}\}$ and
$D'$ is an optimal $Q'$-\stds \ of $H'$ then the set $D$ returned 
by \Domination$(H,Q)$ 
is an optimal $Q$-\stds \ of $H$.
\end{lemma}
\begin{proof} 
We first prove that the set $D$ returned by \Domination$(H,Q)$ is a 
$Q$-\stds \ of $H$, then we  prove its optimality.
We distinguish two cases according to the fact that a node $v \in V(H)$ 
is a node in $Q$ or not. W.l.o.g. assume that $v \in V_x$, 
for some $x \in [t]$.\\
- \ If $v \in Q$ then $V_x \cap Q\neq \emptyset$ and by the definition of $Q'$ 
we have that $x\in Q'$.
Hence, since $D'$ is a $Q'$-\stds \ of $H'$, there exists $y \in D'$ that is a 
neighbor of $x$ in $H'$. By Algorithm 1 (see line 8) there exists a node $u_y \in V_y\cap D$. Considering that each node in $V_y$ is a neighbor  of each node in $V_x$ (since $(x,y) \in E(H')$), we have that $v$ is dominated by $u\in D$. \\
- \ Let  $v \in V - Q$.  We know that  $D'$ is a  $Q'$-\stds \ of $H'$.
Hence, if either $x \in Q'$ or $x \not \in Q' \cup D'$ 
we can prove, as in the previous case, that 
there exists  $u \in D$ that dominates $v$. 
Assume now that $x \not \in Q'$ and $x \in D'$  (i.e., $x$ can be not dominated in $H'$).
By the definition of $Q'$ we have that $V_x \cap Q=\emptyset$ and  $V_x$ is a clique.
Hence, since by Algorithm 1 (see line 8) there exists a node 
$u_x \in V_x \cap D$, 
we have that $v$ is a neighbor of $u_x \in D$ in the clique $V_x$.

Now, we prove that $D$ is an optimal $Q$-\stds \ of $H$
whenever $D'$ is an optimal $Q'$-\stds \  of $H'$.
By contradiction, assume that $D$ is not optimal and let $\tD$ be an 
optimal $Q$-\stds \  of $H$.
By Lemma \ref{only-one} we can assume that, for each $x \in [t]$, 
at most one node in $V_x$ is a node in $\tD$. 
Let $\tD' = \{x \ | \ V_x \cap \tD \neq \emptyset\}$.
We claim that $\tD'$ is a $Q'$-\stds \  of $H'$. Indeed:\\
-- \  If  $x \in \tD'$ then there is a node $u \in V_x \cap \tD$.
Since $\tD$ is a $Q$-\stds \  of $H$, we have that there exists 
a node $v \in V_y \cap \tD$  that is a neighbor of $u$ in $H$, 
for some $y \in [t]-\{x\}$.
Hence, $y \in \tD'$ and $y$ is a neighbor of $x$ in $H'$.\\
-- \ If  $x \not\in \tD'$ then $ V_x \cap \tD = \emptyset$.
Hence, any node $u \in V_x$ is dominated by some node $v\in \tD$.
Since $v \in V_y$, for some $y \in [t]-\{x\}$, we have that 
$x$ is a neighbor of $y$ in $H'$; furthermore, $\tD \cap V_y \neq \emptyset$ and so 
$y \in \tD'$. \\
Finally, we prove that $|\tD'| < |D'|$ thus obtaining a contradiction since $D'$ is optimal.\\
By Lemma \ref{only-one} and the construction of $\tD'$, there is a one-to-one
correspondence  between $\tD'$ and $\tD$. Furthermore, by Algorithm 1 
there is a correspondence one-to-one between $D'$ and $D$.
Hence, $|\tD'| = |\tD| < |D| = |D'|$.
\end{proof} 

\begin{theorem}\label{teo-dom}
\Domination$(G,\emptyset)$ returns a minimum dominating set 
in time \com{2^{\dnd(G)}}.
\end{theorem}

\begin{proof}{}
Let $\Ho=G, \H1, \cdots, \Hd$ be the type graph sequence of $G$.
When \Domination$(G,\emptyset)$ is called, Algorithm 1 proceeds recursively, 
and at the $i$-th recursive step, for $i=0,\cdots, d$,  the algorithm is 
called with input graph $\Hii$ and  input node set $Q_i \subseteq V(\Hii)$,
where $Q_i$ is constructed at line 3 of the previous step $i-1$, for $i=1,\cdots,d$, and it is the empty set when $i=0$, i.e., $Q_0=\emptyset$.
 
At step $d$ the algorithm establishes by brute force the optimal $Q_d$-\stds \  of the base graph $\Hd$.

By Lemma \ref{recursive-step}, the set returned at the end of each 
recursive step $i$, for $i=d-1,\cdots, 0$, is the optimal $Q_i$-\stds \ 
 of $\Hii$. Hence, at the end (when $i=0$) the returned set is 
the optimal $\emptyset$-\stds \  of $\Ho$, that by the definition
 is the minimum dominating set of $G$.

Considering that $|V(\Hd)|=\dnd(G)$, the brute search of the solution set at step $d$ requires time $O(2^{\dnd(G)})$. Furthermore, since the construction of the type partition 
of $\Hii$  and of its type graph can be done in polynomial time, and that both 
the construction of $Q_i$ and the selection of the nodes in the solution set 
are easily obtained in linear time, we have \com{2^{\dnd(G)}} time.

\end{proof}

\subsection{Vertex coloring}

In the following we deal with  a generalization of the vertex coloring known as multicoloring.
\begin{definition}
Given a graph $G = (V,E)$ and a weight function $w : V (G) \rightarrow N$, 
the $w$-{\em multicoloring} of $G$  is a function $C$ that assigns 
to each node $v \in V(G)$ a set of $w(v)$ distinct colors such that 
if $(u,v) \in E(G)$ then $C(u) \cap C(v) = \emptyset$. 
 The objective of $w$-multicoloring problem is to
minimize the total number of colors used by the assignment $C$.
\end{definition}
In  case of unitary weights, the    multicoloring problem becomes the   vertex coloring problem.

In the following we  say that a set of colors $C(v)$ assigned to 
a node $v\in V(G)$ is {\em safe} for $v$ if  $C(u) \cap C(v) = \emptyset$ 
whenever $(u,v) \in E(G)$.
\begin{lemma} \label{independent}
Let $G$ be a graph and let  $w : V (G) \rightarrow N$ be a weight function.
Let  $V_1,\cdots, V_t$ be the type partition of $G$.
There exists an optimal $w$-multicoloring $C$ of $G$ such that 
for each independent set $V_x$ in the type partition of $G$
it holds $C(u) \subseteq C(v)$  for $u \in V_x$ and
$v=\argmax_{u \in V_x} |C(u)|$.
\end{lemma}
		
\begin{proof}{}
Let $V_x$ be any independent set and let $v_x=\argmax_{u \in V_x} |C(u)|$.
Since $v_x$ shares with any other node $u\in V_x$ the neighborhood, we have 
that the set of colors $C(v_x)$ is safe for each node $u \in V_x$. Hence, 
$C(u)\subseteq C(v_x)$ with $|C(u)|=w(u)$ is safe for $u$.
This allows to define a new optimal $w$-multicoloring $C'$ as follows:
$$C'(u)= \begin{cases}
        C(u) & \mbox{if $u \in V_x$ and $V_x$ is a clique}\\
			 C(v_x)  & \mbox{if $u=v_x$ and $V_x$ is an independent set}\\
	\mbox{$C(u)\subseteq C(v_x)$ with $|C(u)|=w(u)$}  & 
	\mbox{if $u\in V_x-\{v_x\}$ and $V_x$ is an independent set.} 
																						\end{cases}$$
\end{proof}

The proposed  FPT algorithm \Coloring\  is given as Algorithm \ref{alg3}. 
Let  $w_u$ be a unitary weight function, that is, $w_u : V (G) \rightarrow \{1\}$.
Initially the algorithm is called with \Coloring($G,w_u$).
The algorithm recursively constructs the graphs in the type graph 
sequence of $G$, until the   base graph is obtained.
At each recursive step, the algorithm has as input a graph $H$ and the weight function $w$  that, for each node $u\in V(H)$, gives  the number 
$w(u)$ of colors that must  be assigned to $u$. 
The weight function $w$ is obtained from the one at the previous step.
In case $H$ is the base graph, then
the algorithm solves the ILP shown at the lines 4-5 and, by following \cite{L},   obtains the minimum $w$-multicoloring $C$ of $H$.
If, otherwise, $H$ is not the base graph then the algorithm 
first constructs the type graph $H'$ of $H$ and opportunely evaluates 
the weight $w'(x)$ to be assigned to each node $x$ in $V(H')$,
then it uses the $w'$-multicoloring
 $C'$ of $H'$  returned by \Coloring($H',w'$)
to build  and return the $w$-multicoloring $C$ of $H$. 
In particular, the algorithm
considers each set $V_x$ in the type partition of $H$ and
distributes the $w'(x)$ colors assigned to $x$ to the nodes in $V_x$
taking account of the fact whether  $V_x$ is a clique or an independent set. 
\begin{algorithm}
\SetCommentSty{footnotesize}
\SetKwInput{KwData}{Input}
\SetKwInput{KwResult}{Output}
\DontPrintSemicolon
\caption{ \ \   \textbf{Algorithm} \Coloring($H, w$) \label{alg3}}
\KwData { A graph $H=(V(H),E(H))$, a weighted function $w: V(H) \rightarrow N_0$ of }
\setcounter{AlgoLine}{0}

\uIf{ $H$ a base graph }{ 
Let $\I$ be the set of all independent sets of nodes in $H$. \\
Solve the following ILP: \\
\qquad \qquad $\min \sum_{I \in \I} z_I$\\
\qquad \qquad  $ \sum_{I : u\in I} z_I =w(u)$ for each $u \in V(H)$\\
\BlankLine
 \For{\textbf{each} $I$ such that $z_I > 0$}{
      Choose $z_I$ new colors and
		 give to each   $u \in I$ a set $C(u)$ of $w(u)$ such colors, 
			}
 }
\uElse{ 
Let $V_1,\cdots, V_t$ be the type partition of $H$ and  $H'$ be the type graph of $H$. \\
 \lFor{\textbf{each} $x \in V(H')$}{$w'(x)= \begin{cases}
                      \sum_{u \in V_x}w(u) & \mbox{if $V_x$ is a clique}\\
											\max_{u \in V_x} w(u) & \mbox{if $V_x$ is an independent set}
											\end{cases}$
											} 
$C'$ = \Coloring$(H',w')$ \\
\BlankLine
 \For{\textbf{each} $x \in V(H')$}{
    \uIf{ $V_x$ is a clique }{ 
		 Let $u_1,u_2,\cdots,u_{|V_x|}$ be the nodes in $V_x$ \\
		 Partition  $C'(x)$ in $|V_x|$ subsets, 
		   $C(u_1), C(u_2), \cdots, C(u_{|V_x|})$, s.t. $|C(u_i)|=w(u_i)$ for $i=1,2,\cdots,|V_x|$\\
					}
						\lElse{ Assign to each node $u \in V_x$ a set $C(u) \subseteq C'(x)$, 
			with $|C(u)|=w(u)$}                  
}
}
 \Return $C$
\end{algorithm}

\remove{
If $V_x$ is a clique then the $w'(x)$ colors are partitioned assigning 
$w(u)$ colors to $u$, for each $u \in V_x$ (recall that  by line 10 it holds 
$w'(x)=\sum_{u \in V_x}w(u)$). If  $V_x$ is an independent set then any subset of
$w(u)$ colors among the available $w'(x)$ is assigned to $u$, for $u \in V_x$ 
(recall that by line 10 it holds $w'(x)=\max_{u \in V_x} w(u)$). 
}
\begin{lemma} \label{recursive-step-coloring}
Let $H$ be not a base graph and let  $w : V (G) \rightarrow N$ be a weight function.
Let $V_1,\cdots, V_t$ be the type partition of $H$ and let $H'$ be its type graph.
If \\
\centerline{$w'(x)= \begin{cases}
      \sum_{u \in V_x}w(u) & \mbox{if $V_x$ is a clique}\\
			\max_{u \in V_x} w(u) & \mbox{if $V_x$ is an independent set}
		\end{cases}$}\\
		and
$C'$ is an optimal $w'$-multicoloring of $H'$ then the coloring $C$ returned 
by \Coloring$(H,w)$ 
is an optimal $w$--multicoloring  of $H$.
\end{lemma}

\begin{proof}{}
We first prove that the coloring $C$ returned by \Coloring$(H,w)$ is a 
$w$-multicoloring  of $H$, then we will prove that it is also optimal.
Let $v$ be any node in $V(H)$.
W.l.o.g. assume that $v \in V_x$, for some $x \in [t]$.
We distinguish two cases according to the fact that $V_x$ is a clique 
or an independent set.\\
- Let $V_x$ be a clique. Since by the hypothesis $w'(x)=\sum_{u \in V_x}w(u)$ 
then the partitioning of the colors in $C'(x)$ in the $|V_x|$ sets 
given at line 15 correctly assigns a set $C(u)$ of $w(u)$ colors to $u \in V_x$.
Now we claim that $C(u)$ is safe for $u$, i.e.,
 $C(u) \cap C(v) =\emptyset$ whenever $(u,v)\in E(H)$.
If $u,v \in V_x$ then the claim follows since $C(u)$ and $C(v)$ are two sets in the partition of $C'(x)$. If $u \in V_x$ and $v \in V_y$, for $y\neq x$, then by the construction of the type graph $H'$ we have $(x,y)\in E(H')$ and, by the hypothesis, $C'(x) \cap C'(y) = \emptyset$. Considering that $C(v) \subseteq C'(y)$, we have the claim also in this case.\\
- Let $V_x$ be an independent set. Let $u \in V_x$. 
Since $w'(x)=\max_{u \in V_x} w(u)$ then the algorithm (at line 16) correctly assigns a set $C(u) \subseteq C'(x)$ of $w(u)$ colors to $u$.  
Since $V_x$ is an independent set we have that 
each neighbor $v$ of $u$ in $H$ is a node of some 
$V_y$ with $y\neq x$ and $(x,y) \in E(H')$. 
As in the above case, it is easy to see that
since $C'(x) \cap C'(y) = \emptyset$ and $C(v) \subseteq C'(y)$ we have 
$C(u) \cap C(v) =\emptyset$; then, $C(u)$ is safe for $u$. 

Now, we prove that $C$ is an optimal $w$-multicoloring of $H$
whenever $C'$ is an optimal $w'$-multicoloring of $H'$.
By contradiction, assume that $C$ is not optimal and let $\tC$ be an 
optimal $w$-multicoloring of $H$, that is 
$|\bigcup_{u\in V(H)}\tC(u)| < |\bigcup_{u\in V(H)}C(u)|$.
We assume that $\tC$ is the optimal multicoloring pinpointed by
 Lemma \ref{independent}.
By using $\tC$ we define a $w'$-multicoloring $\tC'$ of $H'$ as follows:
\begin{equation} \label{new-color}
\tC'(x)=\bigcup_{u \in V_x}\tC(u) \qquad \qquad \mbox{for $x \in V(H')$}.
\end{equation}
We first prove that $|\tC'(x)|=w'(x)$, then we prove that $\tC'(x)$ is safe for $x$. Finally, we prove that the number of colors used by 
$\tC'$ is less than the number of colors used by $C'$, thus obtaining a contradiction. \\
If $V_x$ is a clique then $\tC(u) \cap \tC(v)=\emptyset$ for each pair 
of nodes $u,v \in V_x$. Hence $|\tC'(x)|=\sum_{u \in V_x}|\tC(u)|=
\sum_{u \in V_x}w(u)=w'(x)$. 
If $V_x$ is an independent set then by using Lemma \ref{independent}
we have that if $v=\argmax_{u \in V_x} |\tC(u)|$ then 
$\tC(u) \subseteq \tC(v)$  for $u \in V_x$. Hence, 
$|\tC'(x)|=|\bigcup_{u \in V_x}\tC(u)|=|\tC(v)|=w(v)=\max_{u \in V_x} w(u) =w'(x)$.\\
To prove that $\tC'(x)$ is safe for $x$, consider that
for each $y \in V(H')$ with $(x,y) \in E(H')$ we have that 
$(u,v) \in E(H)$ for each $u \in V_x$ and $v \in V_y$. 
Hence, since $\tC(u) \cap \tC(v)=\emptyset$, we have 
$\tC'(x) \cap \tC'(y)=\emptyset$ and then $\tC'(x)$ is safe for $x$.\\
To complete the proof, we recall that the construction of $C$ by $C'$ in Algorithm \ref{alg3} assures that
  $\bigcup_{u \in V(H)}C(u) = \bigcup_{x\in V(H')}C'(x)$. Hence, by (\ref{new-color}) we have
$$\big|\bigcup_{x \in V(H')}\tC'(x) \big| = \big|\bigcup_{x \in V(H')}\bigcup_{u \in V_x}\tC(u) \big|=
\big|\bigcup_{u \in V(H)}\tC(u) \big| <  \big|\bigcup_{u\in V(H)}C(u)\big| = 
 \big| \bigcup_{x\in V(H')}C'(x) \big|$$

\end{proof}

\remove{
****************************
\proof ({\em Sketch.}) 
We first prove that the coloring $C$ returned by \Coloring$(H,w)$ is a 
$w$-multicoloring  of $H$, then we will prove that it is also optimal.
Let $v$ be any node in $V(H)$.
W.l.o.g. assume that $v \in V_x$, for some $x \in [t]$.
We distinguish two cases according to the fact that $V_x$ is a clique 
or an independent set.\\
- Let $V_x$ be a clique. Since by the hypothesis $w'(x)=\sum_{u \in V_x}w(u)$ 
then the partitioning of the colors in $C'(x)$ in the $|V_x|$ sets 
given at line 15 correctly assigns a set $C(u)$ of $w(u)$ colors to $u \in V_x$.
Now we claim that $C(u)$ is safe for $u$, i.e.,
 $C(u) \cap C(v) =\emptyset$ whenever $(u,v)\in E(H)$.
If $u,v \in V_x$ then the claim follows since $C(u)$ and $C(v)$ are two sets in the partition of $C'(x)$. If $u \in V_x$ and $v \in V_y$, for $y\neq x$, then by the construction of the type graph $H'$ we have $(x,y)\in E(H')$ and, by the hypothesis, $C'(x) \cap C'(y) = \emptyset$. Considering that $C(v) \subseteq C'(y)$, we have the claim also in this case.\\
- Let $V_x$ be an independent set. Let $u \in V_x$. 
Since $w'(x)=\max_{u \in V_x} w(u)$ then the algorithm (at line 16) correctly assigns a set $C(u) \subseteq C'(x)$ of $w(u)$ colors to $u$.  
Since $V_x$ is an independent set we have that 
each neighbor $v$ of $u$ in $H$ is a node of some 
$V_y$ with $y\neq x$ and $(x,y) \in E(H')$. 
As in the above case, it is easy to see that
since $C'(x) \cap C'(y) = \emptyset$ and $C(v) \subseteq C'(y)$ we have 
$C(u) \cap C(v) =\emptyset$; then, $C(u)$ is safe for $u$. 

Now, we prove that $C$ is an optimal $w$-multicoloring of $H$
whenever $C'$ is an optimal $w'$-multicoloring of $H'$.
By contradiction, assume that $C$ is not optimal and let $\tC$ be an 
optimal $w$-multicoloring of $H$, that is 
$|\bigcup_{u\in V(H)}\tC(u)| < |\bigcup_{u\in V(H)}C(u)|$.
We assume that $\tC$ is the optimal multicoloring pinpointed by
 Lemma \ref{independent}.
By using $\tC$ we define a $w'$-multicoloring $\tC'$ of $H'$ as follows:
$\tC'(x)=\bigcup_{u \in V_x}\tC(u)$  \ for $x \in V(H')$.
We can prove that $|\tC'(x)|=w'(x)$ and that $\tC'(x)$ is safe for $x$.
 Finally, we can also prove that the number of colors used by 
$\tC'$ is less than the number of colors used by $C'$, thus obtaining a contradiction. 
\qed
*****************
}

\begin{theorem} \label{teo-coloring}
\Coloring$(G,w_u)$ returns a minimum coloring of $G$ in time {\com{t^{2,5 t + o(t)}\log n}}, where $t=\dnd(G)$.
\end{theorem}

\begin{proof}{}
Let $\Ho=G, \H1, \cdots, \Hd$ be the type graph sequence of $G$.
When \Coloring$(G,w_u)$ is called, Algorithm \ref{alg3} proceeds recursively, 
and at the $i$-th recursive step, for $i=0,\cdots, d$,  the algorithm is 
called with input graph $\Hii$ and  input weighted function $w_i$,
where $w_i$ is constructed at line 10 of the previous step $i-1$, for $i=1,\cdots,d$, and it is the unitary weighted function when $i=0$, i.e., $w_0=w_u$.
 
At step $d$ the algorithm solves an ILP that generalizes the ILP introduced by  Lampis in \cite{L} to obtain an FPT algorithm for proper coloring 
the nodes of a graph. 
Indeed, considering that to guarantee the safety of a multicoloring, 
each color class consists of an independent set of nodes in $\Hd$, 
the ILP at lines 4-5 uses the set $\I$ of all the independent sets of 
nodes in $\Hd$ and determines the number $z_I$, for $I \in \I$, 
of colors to be assigned to the nodes in $I$.
The target is to minimize the total number of used colors, i.e. $\sum_{I \in \I} z_I$, subject to the following constraints: For each node $u \in V(\Hd)$, the sum of the number of colors $z_I$ assigned to each independent set $I$ who $u$ belongs to is exactly equal to the number of colors that $u$ needs, i.e., $w_d(u)$.
Hence, the assignment $C(u)$  to node $u \in I$ (see line 7) of $w_d(u)$ 
colors chosen among the $z_I$ colors assigned to $I$,  
is an optimal $w_d$-multicoloring of the base graph $\Hd$.

By Lemma \ref{recursive-step-coloring}, the multicoloring returned 
at the end of each recursive step $i$, for $i=d-1,\cdots, 0$, is the 
optimal $w_i$-multicoloring of $\Hii$. Hence, at the end (when $i=0$) the returned multicoloring is the optimal $w_u$-multicoloring of $\Ho$, that by the definition is the minimum coloring of $G$.

{ To evaluate the time of our algorithm we use 
the well-known result that Integer Linear Programming is fixed parameter 
tractable parameterized by the number of variables.
\begin{itemize}
	\item[] \textbf{$t$-Variable Opt Integer Linear Programming}
\\
\textbf{Instance:}  A matrix $A \in Z^{m\times t}$ and vector $b \in Z^m$ and  $c \in Z^t$.
\\
\textbf{Question:} Find a vector $x \in Z^t$ that minimize $c^\top x$ and satisfies $Ax \geq b$?
\end{itemize}
\begin{theorem} \label{tprop}\cite{FLMRS}
$t$-Variable Opt Integer Linear Programming can be solved in time $O(t^{2.5t+o(t)} \cdot L \cdot \log(MN))$ where $L$ is the number of bits in the input  $N$ is
the maximum absolute values any variable can take, and $M$ is an upper bound on the absolute value of the minimum taken by the objective function.
\end{theorem}
}

{Since $|V(\Hd)|=\dnd(G)$, the ILP at lines 4-5 uses $2^{\dnd(G)}$ 
variables and $\dnd(G)$ constraints. As highlighted by Lenstra (see section 4 in \cite{Len}), such ILP can 
be reduced to an ILP with only $\min\{\dnd(G),2^{\dnd(G)}\} = \dnd(G)$ variables.
By Proposition \ref{tprop} we have that it can be solved in time $O(t^{2.5t+o(t)} \log n)$ where $t=\dnd(G)$.
Furthermore, since the construction of the type partition 
of $\Hii$  and of its type graph can be done in polynomial time, and that both 
the construction of $w_i$ and the selection of the colors for each node 
$u \in V(\Hii)$ 
are easily obtained in linear time, we have \com{t^{2.5t+o(t)} \log n} time.}
\end{proof}

{  An algorithm parameterized by modular-width, which obtains the 
minimum number of colors to color properly a graph $G$ was presented in \cite{GLO}}. We stress that,  {a part simplicity and efficiency questions,}  such an algorithm does not provide the coloring of the vertices.

\def\VC{{{\sf Vertex Cover}}}

\subsection{Vertex cover}
\begin{algorithm}[tb!]
\SetCommentSty{footnotesize}
\SetKwInput{KwData}{Input}
\SetKwInput{KwResult}{Output}
\DontPrintSemicolon
\caption{ \ \    \VC($H, w,s$) \label{algvc}}
\KwData { A graph $H=(V(H),E(H))$, 	two weighted functions $w: V(H) \rightarrow N_0$, $s: V(H) \rightarrow N_0$  }
\setcounter{AlgoLine}{0}
\uIf{ $H$ is a base graph }{ 
$C=V(H)$\\
\For{\textbf{each} $S\subseteq V(H)$}{ 
	\lIf{ ($S$ is a vertex cover of $H$) and ($\sum_{v\in S} s(v)+\sum_{v \notin S} w (v) <\sum_{v\in C} s(v)+\sum_{v \notin C} w (v)$)}{
			$C=S$
		}
	}  
}
\uElse{ 
Let $V_1,\cdots, V_t$ be the type partition of $H$and $H'$  the type graph of $H$. \\
\lFor{ $x \in V(H')$}{$w'(x)= \begin{cases}
                     \min_{v\in V_x}\left( w(v) + \sum_{u \in V_x\atop u\neq v} s(u)\right) & \mbox{if $V_x$ is a clique}\\
											\sum_{u \in V_x} w(u) & \mbox{otherwise}
											\end{cases}$
											} 								
 \lFor{ $x \in V(H')$}{$s'(x)= \sum_{u \in V_x} s(u)$} 								
$C'$ = \VC$(H',w',s')$ \\
$C=\emptyset$ \\
 \For{\textbf{each} $x \in V(H')$}{
            \lIf{$x \in C'$}{ $C=C\cup V_x$}
						\lElseIf{ $V_x$ is a clique}{ \footnotesize{$v_x = \argmax_{u \in V_x} (s(u)-w(u))$;
						   $C=C \cup (V_x-\{v_x\})$}}
					}
} 
\Return $C$
\end{algorithm}

We consider  the following generalization of the weighted vertex cover.
\begin{definition}
Given a graph $G=(V,E)$ and two weight functions  $w:V \rightarrow N$ and  $s:V \rightarrow N$ s.t.   $w(v)\leq s(v)$  for each $v \in V$,
the {\em 2-Weighted Vertex Cover (2-WVC)} of $G$ respect to $s(\cdot)$ 
and $w(\cdot)$ is a set $C \subseteq V$  s.t.
$C$ is a vertex cover  for $G$, which minimizes the value
 $Cost(C)=\sum_{v\in C} s(v)+\sum_{v \notin C} w (v)$.
\end{definition}
When $w(v)=0$ and $s(v)=1$
for each $v \in V$,  a 2-WVC of $G$  is a vertex cover of $G$.  

Algorithm \ref{algvc} shows the FPT algorithm \VC. 
The algorithm recursively constructs graphs in the type graph
sequence of $G$, until the base graph is obtained.
It is initially called with \VC($G,w,s$), where for each $v \in V$ we have $w(v)=0$ and $s(v)=1$.
Intuitively the function $s(\cdot)$ recursively counts the number of nodes  of $G$ that are represented by  a metavertex, while the function $w(\cdot)$ computes the minimum number of nodes of $G$ needed to cover the internal edges of a metavertex.
At each recursive step, the algorithm takes as input a graph $H$ and the two functions $s(\cdot)$ and  $w(\cdot)$ computed in the previous step.
The goal of the algorithm is to compute for each  $H$ in the type graph sequence, a subset $C\subseteq V(H)$ of nodes that is
 a $2$-WVC of $H$.
 Hence, in order to show that the algorithm is correct, we need  to prove that given $C' \subseteq V(H')$ that is a $2$-WVC of $H'$, 
where $H'$ is the type graph of $H$, the solution $C\subseteq V(H)$ -- computed by the algorithm for $H$  --  is a $2$-WVC of $H$. {The result will follow since, in the initial graph $G$,  for each $v \in V$ we have $w(v)=0$ and $s(v)=1$ and consequently  the $2$-WVC problem corresponds to the minimum vertex cover problem.}  
\begin{lemma} \label{lemmaVC}
Let $(H,w,s)$ be an instance of the 2-WVC problem, where $H$ is not a base graph. Let $H'$ be the type graph of $H$ and let $w'$ and $s'$ be the weight functions of $V(H')$ computed by Algorithm \ref{algvc}. If $C'\subseteq V(H')$ is an optimal solution for  $(H',w',s')$, then the solution $C\subseteq V$ computed by  the Algorithm \ref{algvc} is an optimal solution for  2-WVC  of $(H,w,s)$. 
\end {lemma} 

\begin{proof}
Let $K',I'$ be a partition of $V(H')$ such that for each $x \in K'$ it holds $V_x$ is a clique, and for each $x \in I'$ it holds $V_x$ is an independent set. 
It is worth to observe that, 
by construction (see lines 12-13), if $C'$ is a vertex cover of $H'$ 
then $C$ is a vertex cover of $H$ and that
\begin{eqnarray*}
Cost(C) & = & \sum_{v\in C} s(v)+\sum_{v \notin C} w (v) \\
 & = & \sum_{x\in C'} \sum_{v\in V_x} s(v) + 
        \sum_{x \in K'-C'} \sum_{v\in V_x-\{v_x\}} s(v) +
				 \sum_{x \in I'-C'}\sum_{v\in V_x} w(v)  
				   \qquad \mbox{(by lines 12-13)}\\
 & = & \sum_{x\in C'} s'(x) + 
       \sum_{x \in K'-C'}	\left( \sum_{v\in V_x} s(v)  - s(v_x)  
					\right) + \sum_{x \in I'-C'} w' (v) 
			      \qquad \mbox{(by lines 7-8)}\\		
 & = & \sum_{x\in C'} s'(x) + 
       \sum_{x \in K'-C'}	\left( \sum_{v\in V_x} s(v)  
			    - \max_{u\in V_x}(s(u)-w(u)) \right)   + \sum_{x \in I'-C'} w' (v) 
			      \qquad \mbox{(by line 13)}\\
 & = & \sum_{x\in C'} s'(x) + 
       \sum_{x \in K'-C'}w'(x)   + \sum_{x \in I'-C'} w' (v) 
			      \qquad \mbox{(by line 7)}\\												
					 & = & \sum_{v\in C'} s(v)+\sum_{v \notin C'} w (v) = Cost(C').
\end{eqnarray*}

By contradiction, let $C'\subseteq V(H')$ be an optimal solution for $(H',w',s')$ while the solution $C$ computed by Algorithm \ref{alg3} 
(lines 11-13) is not optimal.
Then there exists $C^*\subseteq V(H)$ such that $C^*$ is a vertex cover of $H$ and $Cost(C^*)<Cost(C)=Cost(C').$
\\
Let  
\begin{equation*}C''= \{x \ |\ x\in V(H') \mbox{ and } V_x\subseteq C^*\}.
\end{equation*}
Now, we first prove that $C''$ is a vertex cover of $H'$, then we show that
 $Cost(C'')<Cost(C')$  contradicting the optimality of the vertex cover $C'$.\\
Assume that $C''$ is not a vertex cover. Hence, 
there exists $(x,y) \in E(H')$ such that $x\not \in C''$ and $y\not \in C''$, and, by the definition of $C''$, we have that there exist 
$u \in V_x-C^*$ and $v \in V_y-C^*$. Furthermore, since  $(x,y) \in E(H')$ then there is a complete bipartite graph between $V_x$ and $V_y$ in $H$; hence, $(u,v) \in E(H)$. Hence, $u,v \not \in C^*$,  and 
 $(u,v) \in E(H)$ contradicting the assumption that  $C^*$ is a vertex cover of $H$.
\\
We now  prove that $Cost(C'')\leq Cost(C^*)$, from which it follows that
\[Cost(C'')\leq Cost(C^*)<Cost(C)=Cost(C'),\] 
thus concluding the proof.
\begin{eqnarray} \label{eqdiff}
\nonumber \lefteqn{Cost(C^*)-Cost(C'')=\sum_{v\in C^*} s(v) +\sum_{v\notin C^*} w(v) - \sum_{x\in C''} s'(x) -\sum_{x\notin C''} w'(x) } \\
\nonumber 
&=& \left(\sum_{x\in C''}\sum_{v\in V_x } s(v)+\sum_{x\notin C''}\sum_{v\in V_x \cap C^*} s(v)\right)+ 
\sum_{x\notin C''}\sum_{v\in V_x-C^*} w(v) - \sum_{x\in C''} s'(x) -\sum_{x\notin C''} w'(x)\\
\nonumber 
&=& \left(\sum_{x\in C''}s'(x)+\sum_{x\notin C''}\sum_{v\in V_x \cap C^*} s(v)\right)+ 
\sum_{x\notin C''}\sum_{v\in V_x-C^*} w(v) - \sum_{x\in C''} s'(x) -\sum_{x\notin C''} w'(x) \ \ \ \mbox{(by line 9) }\\ \nonumber
&=& \sum_{x\notin C''}\sum_{v\in V_x \cap C^*} s(v)+ 
\sum_{x\notin C''}\sum_{v\in V_x-C^*} w(v)  -\sum_{x\notin C''} w'(x) \\ \nonumber
&=& \sum_{x\notin C''}\left (\sum_{v\in V_x \cap C^*} s(v)+ 
\sum_{v\in V_x-C^*} w(v)  -w'(x) 
\right )\\ \label{aa}
&=& \sum_{x\in I'-C''}\left (\sum_{v\in V_x \cap C^*} s(v)+ 
\sum_{v\in V_x-C^*} w(v)  - \sum_{v \in V_x} w(v)
\right )\\ \nonumber
& & + \sum_{x\in K'-C''}\left (\sum_{v\in V_x \cap C^*} s(v)+ 
\sum_{v\in V_x-C^*} w(v)  - \sum_{v \in V_x} s(v) +(s(v_x)-w(v_x)) 
\right ).
\end{eqnarray}

\noindent
Since by construction $s(v) \geq w(v)$, it holds
\begin{equation} \label{ind}
\sum_{v\in V_x \cap C^*} s(v)+ 
\sum_{v\in V_x-C^*} w(v)  - \sum_{v \in V_x} w(v) \geq 0.
\end{equation}
Recalling that $C^*$ is a vertex cover of $H$, we have that  
for each clique metavertex $V_x$, $C^*$ must contain at least $|V_x|-1$ of its nodes (otherwise $H$ contains at least a non covered edge). 
Knowing that   $s(v_x) \geq s(v)$ (cfr.  line 14 in Algorithm 2), it follows 
\begin{equation} \label{cli}
\sum_{v\in V_x \cap C^*} s(v) -  (\sum_{v \in V_x} s(v) -s(v_x)) \geq 0.
\end{equation}
By (\ref{ind}) and (\ref{cli}) and recalling (\ref{aa}) we have 
$Cost(C^*)-Cost(C'') \geq 0$.
\end{proof}

\remove{
**************

\proof ({\em Sketch.})
Let $K',I'$ be a partition of $V(H')$ such that for each $x \in K'$ it holds $V_x$ is a clique, and for each $x \in I'$ it holds $V_x$ is an independent set. 
It is worth to observe that, 
by construction (see lines 12-13), if $C'$ is a vertex cover of $H'$ 
then $C$ is a vertex cover of $H$ and that $Cost(C)=Cost(C')$.
\remove{
\begin{eqnarray*}
Cost(C) & = & \sum_{v\in C} s(v)+\sum_{v \notin C} w (v) \\
 & = & \sum_{x\in C'} \sum_{v\in V_x} s(v) + 
        \sum_{x \in K'-C'} \sum_{v\in V_x-\{v_x\}} s(v) +
				 \sum_{x \in I'-C'}\sum_{v\in V_x} w(v)  
				   \qquad \mbox{(by lines 12-13)}\\
 & = & \sum_{x\in C'} s'(x) + 
       \sum_{x \in K'-C'}	\left( \sum_{v\in V_x} s(v)  - s(v_x)  
					\right) + \sum_{x \in I'-C'} w' (v) 
			      \qquad \mbox{(by lines 7-8)}\\		
 & = & \sum_{x\in C'} s'(x) + 
       \sum_{x \in K'-C'}	\left( \sum_{v\in V_x} s(v)  
			    - \max_{u\in V_x}(s(u)-w(u)) \right)   + \sum_{x \in I'-C'} w' (v) 
			      \qquad \mbox{(by line 13)}\\
 & = & \sum_{x\in C'} s'(x) + 
       \sum_{x \in K'-C'}w'(x)   + \sum_{x \in I'-C'} w' (v) 
			      \qquad \mbox{(by line 7)}\\												
					 & = & \sum_{v\in C'} s(v)+\sum_{v \notin C'} w (v) = Cost(C').
\end{eqnarray*}
}

By contradiction, let $C'\subseteq V(H')$ be an optimal solution for $(H',w',s')$ while the solution $C$ computed by Algorithm \ref{alg3} 
(lines 11-13) is not optimal.
Then there exists $C^*\subseteq V(H)$ such that $C^*$ is a vertex cover of $H$ and $Cost(C^*)<Cost(C)=Cost(C').$
Let  
\\
\centerline{$C''= \{x \ |\ x \in V(H') \mbox{ and } V_x\subseteq C^*\}.$}
\\
We first prove that $C''$ is a vertex cover of $H'$, then we show that
 $Cost(C'')<Cost(C')$  contradicting the optimality of the vertex cover $C'$.\\
Assume that $C''$ is not a vertex cover. Then, 
there exists $(x,y) \in E(H')$ such that $x\not \in C''$ and $y\not \in C''$, and, by the definition of $C''$, we have that there exist 
$u \in V_x-C^*$ and $v \in V_y-C^*$. Furthermore, since  $(x,y) \in E(H')$ then there is a complete bipartite graph between $V_x$ and $V_y$ in $H$; and  $(u,v) \in E(H)$. Hence, $u,v \not \in C^*$,  and 
 $(u,v) \in E(H)$ contradicting the assumption that  $C^*$ is a vertex cover of $H$.
\\
Since we can  prove that $Cost(C'')\leq Cost(C^*)$, we have \\ 
\hspace*{3truecm}$  Cost(C'')\leq Cost(C^*)<Cost(C)=Cost(C'),$\\ 
thus concluding the proof.
\remove{
\begin{eqnarray} \label{eqdiff}
\nonumber \lefteqn{Cost(C^*)-Cost(C'')=\sum_{v\in C^*} s(v) +\sum_{v\notin C^*} w(v) - \sum_{x\in C''} s'(x) -\sum_{x\notin C''} w'(x) } \\
\nonumber 
&=& \left(\sum_{x\in C''}\sum_{v\in V_x } s(v)+\sum_{x\notin C''}\sum_{v\in V_x \cap C^*} s(v)\right)+ 
\sum_{x\notin C''}\sum_{v\in V_x-C^*} w(v) - \sum_{x\in C''} s'(x) -\sum_{x\notin C''} w'(x)\\
\nonumber 
&=& \left(\sum_{x\in C''}s'(x)+\sum_{x\notin C''}\sum_{v\in V_x \cap C^*} s(v)\right)+ 
\sum_{x\notin C''}\sum_{v\in V_x-C^*} w(v) - \sum_{x\in C''} s'(x) -\sum_{x\notin C''} w'(x) \ \ \ \mbox{(by line 9) }\\ \nonumber
&=& \sum_{x\notin C''}\sum_{v\in V_x \cap C^*} s(v)+ 
\sum_{x\notin C''}\sum_{v\in V_x-C^*} w(v)  -\sum_{x\notin C''} w'(x) \\ \nonumber
&=& \sum_{x\notin C''}\left (\sum_{v\in V_x \cap C^*} s(v)+ 
\sum_{v\in V_x-C^*} w(v)  -w'(x) 
\right )\\ \label{aa}
&=& \sum_{x\in I'-C''}\left (\sum_{v\in V_x \cap C^*} s(v)+ 
\sum_{v\in V_x-C^*} w(v)  - \sum_{v \in V_x} w(v)
\right )\\ \nonumber
& & + \sum_{x\in K'-C''}\left (\sum_{v\in V_x \cap C^*} s(v)+ 
\sum_{v\in V_x-C^*} w(v)  - \sum_{v \in V_x} s(v) +(s(v_x)-w(v_x)) 
\right ).
\end{eqnarray}

\noindent
Since by construction $s(v) \geq w(v)$, it holds
\begin{equation} \label{ind}
\sum_{v\in V_x \cap C^*} s(v)+ 
\sum_{v\in V_x-C^*} w(v)  - \sum_{v \in V_x} w(v) \geq 0.
\end{equation}
Recalling that $C^*$ is a vertex cover of $H$, we have that  
for each clique metavertex $V_x$, $C^*$ must contain at least $|V_x|-1$ of its nodes (otherwise $H$ contains at least a non covered edge). 
Knowing that   $s(v_x) \geq s(v)$ (cfr.  line 14 in Algorithm 2), it follows 
\begin{equation} \label{cli}
\sum_{v\in V_x \cap C^*} s(v) -  (\sum_{v \in V_x} s(v) -s(v_x)) \geq 0.
\end{equation}
By (\ref{ind}) and (\ref{cli}) and recalling (\ref{aa}) we have 
$Cost(C^*)-Cost(C'') \geq 0$.
}
\qed
***********************
}

\begin{theorem}\label{teoVC}
\VC$(G,w,s)$ where for each $v\in V(G),$ $w(v)=0$ and $s(v)=1$ returns the minimum vertex cover of $G$ in time \com{2^{\dnd(G)}}.
\end{theorem}

\begin{proof}
The algorithm recursively constructs graphs in the type graph sequence 
of $G$, until a base graph is obtained.
When $H$ is a base graph then \VC$(H,w,s)$ searches by brute force 
the set $C$ that is the 2-WVC respect to $s(\cdot)$ and $w(\cdot)$, 
as computed by the algorithm, and returns it. 
\end{proof}

\section{Conclusion} 
We introduced a novel parameter, named iterated type partition  and  examined some of its properties. We show that the Equitable Coloring problem is W[1]-hard when parametrized by the
iterated type partition. This result extends also to the  modular-width parameter. We also prove that the hardness drops for the neighborhood diversity parameter, when the problem becomes FPT. Moreover, we presented a  general strategy that enables to find FPT algorithms for several problems when parameterized by iterated type partition.  {Algorithms for Dominating set, Vertex coloring and Vertex cover problems have been presented, while the algorithms for Clique and Independent set  problems will appear in the extended version of the work.}
It would be interesting to investigate if the proposed strategy can be applied on other problems. 
As a direction for future research, it would be interesting to analyze the Edge dominating set problem, which has been shown to be FPT with the neighborhood diversity parameter \cite{L}.

\end{document}